\documentclass[11pt,a4paper]{article}
\usepackage[margin=2cm]{geometry}
\usepackage{setspace}

\usepackage [only, mapsfrom]{stmaryrd}

\usepackage{undertilde}
\usepackage{amsfonts}
\usepackage{amssymb}
\usepackage{amsthm}
\usepackage{amsmath}
\usepackage{graphicx}
\usepackage{hyperref}
\usepackage{enumerate}
\usepackage{tikz}
\usetikzlibrary{decorations.markings,arrows,calc}
\usepackage{scalerel}
\usetikzlibrary{svg.path}

\definecolor{orcidlogocol}{HTML}{A6CE39}
\tikzset{
	orcidlogo/.pic={
		\fill[orcidlogocol] svg{M256,128c0,70.7-57.3,128-128,128C57.3,256,0,198.7,0,128C0,57.3,57.3,0,128,0C198.7,0,256,57.3,256,128z};
		\fill[white] svg{M86.3,186.2H70.9V79.1h15.4v48.4V186.2z}
		svg{M108.9,79.1h41.6c39.6,0,57,28.3,57,53.6c0,27.5-21.5,53.6-56.8,53.6h-41.8V79.1z M124.3,172.4h24.5c34.9,0,42.9-26.5,42.9-39.7c0-21.5-13.7-39.7-43.7-39.7h-23.7V172.4z}
		svg{M88.7,56.8c0,5.5-4.5,10.1-10.1,10.1c-5.6,0-10.1-4.6-10.1-10.1c0-5.6,4.5-10.1,10.1-10.1C84.2,46.7,88.7,51.3,88.7,56.8z};
	}
}

\newcommand\orcidicon[1]{\href{https://orcid.org/#1}{\mbox{\scalerel*{
				\begin{tikzpicture}[yscale=-1,transform shape]
					\pic{orcidlogo};
				\end{tikzpicture}
			}{|}}}}

\usepackage{color}
\usepackage[mathscr]{euscript}

\usepackage{leftidx,slashed}

\usepackage{mathrsfs}
\usepackage{hyperref}
\hypersetup{
	breaklinks=true,   
	colorlinks=true,   
	pdfusetitle=true,  
}
\usepackage{datetime}

\numberwithin{equation}{section}

\newcommand{\id}{\mathrm{id}}
\DeclareMathOperator{\pr}{pr}

\newcommand{\bbLbrack}{[\kern-0.4em{[}\,}
\newcommand{\bbRbrack}{\,]\kern-0.4em{]}}

\newcommand{\kb}{\boldsymbol{k}}
\newcommand{\vb}{\boldsymbol{v}}
\newcommand{\ii}{{\rm i}}

\newcommand{\x}{{\rm x}}

\newcommand{\vol}{{\rm vol}}

\newcommand{\dd}{{\rm d}}

\DeclareMathOperator{\supp}{supp}

\DeclareMathOperator{\End}{End}
\DeclareMathOperator{\Hom}{Hom}
\DeclareMathOperator{\Lin}{Lin} 
\DeclareMathOperator{\Char}{Char}

\newcommand{\phg}{\textnormal{phg}}
\newcommand{\cl}{\textnormal{cl}}
\newcommand{\sub}{\textnormal{sub}}

\DeclareMathOperator{\Op}{Op}

\newcommand{\beq}{\begin{equation}}
\newcommand{\ene}{\end{equation}}

\newcommand{\RR}{\mathbb{R}}
\newcommand{\CC}{\mathbb{C}}

\DeclareMathOperator{\WF}{WF}
\newcommand{\pol}{\textnormal{pol}}

\newcommand{\Loc}{\mathsf{Loc}}
 
\newtheorem{thm}{Theorem}[section]
\newtheorem{lemma}[thm]{Lemma}
\newtheorem{cor}[thm]{Corollary}

\theoremstyle{definition}

\newdateformat{daymonthyear}{\THEDAY \, \monthname[\THEMONTH] \THEYEAR}

\newcommand{\DD}{\mathscr{D}}
\newcommand{\EE}{\mathscr{E}}
\renewcommand{\SS}{\mathscr{S}}

\newcommand{\sadj}[1]{\leftidx{^{\star}}{#1}{}}

\newcommand{\dlangle}{\langle\!\langle}
\newcommand{\drangle}{\rangle\!\rangle}

\newcommand{\Nc}{\mathcal{N}}
\newcommand{\Rc}{\mathcal{R}}

\newcommand{\Vc}{\mathcal{V}}

\newcommand{\Had}{\mathrm{Had}}

\newcommand{\rfhgho}{RFHGHO}

\newcommand{\ogth}{{\mathfrak o}}
\newcommand{\tgth}{{\mathfrak t}}

\newcommand{\Mb}{{\boldsymbol{M}}}
\newcommand{\Nb}{{\boldsymbol{N}}}

\newcommand{\knl}{\textnormal{knl}}

\begin{document}

\title{Polarisation sets of Green operators for normally hyperbolic equations}
\author{Christopher J. Fewster{\orcidicon{0000-0001-8915-5321}}\thanks{chris.fewster@york.ac.uk}\\[6pt]  
	\small Department of Mathematics, University of York, Heslington, York YO10 5DD, United Kingdom.\\[4pt]
	\small York Centre for Quantum Technologies, University of York, Heslington, York YO10 5DD, United Kingdom.
} 
 
\date{\daymonthyear\today}

\maketitle
 
\begin{abstract}
	The polarisation set of a vector-valued distribution generalises the wavefront set and captures fibre-directional information about its singularities in addition to their phase space description.
	Motivated by problems in quantum field theory on curved spacetimes, we consider normally hyperbolic operators on vector bundles over globally hyperbolic spacetimes, and 
	compute the polarisation sets of the kernel distributions for their advanced and retarded Green operators and the difference thereof. This permits the computation of related polarisation and wavefront sets for operators whose solution theory is related to the normally hyperbolic case. As a particular example, we consider the Proca equation that describes massive relativistic spin-$1$ particles, identifying and closing a gap in a recent paper on that subject.
\end{abstract}
  
\section{Introduction}
\label{sec:Intro}  

Normally hyperbolic partial differential operators play an important role in the theory of quantum fields on curved spacetimes. Some theories are formulated directly using normally hyperbolic equations, while the solution theory of others relies on relating the field equation to an auxiliary normally hyperbolic operator -- see, e.g.,~\cite{BeniniDappiaggi:2015}. The corresponding Green operators, and their singularity structure, are likewise important. Often it is sufficient to understand the wavefront set of the Green operators, but, as will be described below, there are occasions where a more detailed description of the singularities is required. This will be provided by the present paper, which computes the polarisation set~\cite{Dencker:1982} of the advanced and retarded Green operators, and their difference, for general normally hyperbolic operators on vector bundles over globally hyperbolic spacetimes, and allows the wavefront set of related operators to be computed. This paper is a companion to~\cite{Fewster:2025a,FewsterKlein:2025} in which the present results are applied to discuss the class of Hadamard states for a range of quantum fields, including the Proca equation that models massive spin-$1$ particles such as the $W$ and $Z$ bosons of particle physics. On that subject, this paper will also close a gap
in a recent paper~\cite{MorettiMurroVolpe:2023}, as will be discussed in section~\ref{sec:Proca}.
Let us describe the main results in more detail.

Let $(M,g)$ be a Lorentzian manifold with Levi-Civita connection $\nabla$. Writing $\Omega^\alpha$ for the bundle of weight $\alpha$-densities on $M$, we extend $\nabla$ to act on $\Gamma^\infty(\Omega^\alpha)$ so that $\nabla (-g)^{\alpha/2}=0$, where $(-g)^{1/2}$ is the metric density. Now let $B\to M$ be a finite-rank complex vector bundle over $M$, so that $\Gamma^\infty(B\otimes \Omega^{1/2})$ are the $\tfrac{1}{2}$-densitised smooth sections of $B$. 
By definition, a normally hyperbolic operator on $\Gamma^\infty(B)$ is 
a second-order differential operator $P:\Gamma^\infty(B)\to\Gamma^\infty(B)$ whose principal symbol is given by
$p(x,k)=-g_x^{-1}(k,k)\id_{B_x}$. Thus (recalling the correspondence $k_\mu\leftrightarrow -\ii \partial_\mu$) in a local frame for $B$, $P$ is equivalent to a matrix of differential operators $P^{a}_{\phantom{a}b}$ with  
\begin{equation}
	P^{a}_{\phantom{a}b} = g^{\mu\nu}\nabla_\mu \nabla_\nu \delta^{a}_{\phantom{a}b} + U^{a}_{\phantom{a}b}, 
\end{equation}	
where $U^{a}_{\phantom{a}b}$ is a matrix of first order differential operators. Every
normally hyperbolic operator can be put into a simpler standard form. Namely, there exists a \emph{Weitzenb\"ock connection} $\nabla^B$ on $B$ that gathers all the first order terms 
so that 
\begin{equation}
	P = \Box^B+V, \qquad \Box^B:=g^{\mu\nu}\nabla_\mu^{T^*M\otimes B}\nabla_\nu^B ,
\end{equation}
where $V\in \Gamma^\infty(\End(B))$ while $\nabla^{T^*M\otimes B} = \nabla\otimes 1+1\otimes\nabla^B$ is the natural extension of $\nabla^B$ to a connection on $T^*M\otimes B$. The existence of the Weitzenb\"ock connection $\nabla^B$ is established e.g., in Proposition~3.1 of~\cite{BaumKath:1996} or Lemma~1.5.5 of~\cite{BarGinouxPfaffle}; see also section 1 of~\cite{IslamStrohmaier:2020}. As will be seen in Section~\ref{sec:prop_pol}, the Weitzenb\"ock
connection of a normally hyperbolic operator $P$ on $\Gamma^\infty(B)$ governs the propagation of polarisations for the $\tfrac{1}{2}$-densitised operator
$(-g)^{1/4}P (-g)^{-1/4}$ on $\Gamma^\infty(B\otimes \Omega^{1/2})$.

Now suppose that $(M,g)$ is globally hyperbolic, which means that it is time-oriented, contains no closed causal curves, and has the property that every compact subset of $K$ has a compact causal hull. Then 
(see, e.g.,~\cite{BarGinouxPfaffle}) any normally hyperbolic operator $P$ on $\Gamma^\infty(B)$ has retarded ($+$) and advanced $(-)$ Green operators $E^\pm_P:\Gamma_0^\infty(B)\to\Gamma^\infty(B)$ so that $E_P^\pm Pf=f=PE^\pm_Pf$ for all $f\in \Gamma_0^\infty(B)$ and with support
\begin{equation}
	\supp E^\pm_P f \subset J^\pm(\supp f),
\end{equation}
where $J^\pm(S)$ are the causal future/past of $S\subset M$. (See below for our other main conventions.) The difference of the advanced and retarded operators, $E_P=E^-_P-E^+_P$, plays an important role in quantum field theory. 

It is often convenient to consider operators on bundles with an explicit $\tfrac{1}{2}$-density factor.
If $P$ is normally hyperbolic on $\Gamma^\infty(B\otimes \Omega^{1/2})$, then its Green operators
have distributional kernels in $\DD'((B\boxtimes B^*)\otimes \Omega^{1/2}_{M\times M})$, where $B^*$ is the dual bundle to $B$ and $\boxtimes$ denotes the external bundle product. The wavefront sets of these distributions can be described as follows. 
Let $\Nc\subset T^*M$ be the bundle of nonzero null covectors on $(M,g)$, and, for $(x,k),(y,\ell)\in \Nc$,  write $(x,k)\sim (y,\ell)$ if and only if there is a null geodesic segment $\gamma$ connecting $x$ and $y$, with tangent vectors $k^\sharp$ and $\ell^\sharp$ at the endpoints, which are related by parallel transport along $\gamma$ (understanding $(x,k)\sim (x,\ell)$ if and only if $k=\ell$ is null); in this case we say that the geodesic segment $\gamma$ \emph{witnesses} to the relation $(x,k)\sim (y,\ell)$. With these definitions, set
\begin{align}\label{eq:Rcsets}
	\Rc&= \{(x,k;x',-k')\in \Nc\times\Nc: (x,k)\sim (x',k')\}\\ 
	\Rc^\pm & =\{(x,k;x',-k')\in\Rc:   x\in J^\pm(x') \}.
\end{align}
Then the kernel distributions $E_P^\pm$ and $E_P$ have wavefront sets
	\begin{align}
		\WF(E_P^\pm)&= \Rc^\pm\cup\WF(\id),
		\label{eq:WFEPpm}\\
		\WF(E_P)&= \Rc ,
		\label{eq:WFEP}
	\end{align}
where $\id$ is the kernel distribution of the identity operator on $B$.	
This result generalises the scalar case~\cite{DuiHoer_FIOii:1972,Radzikowski_ulocal1996} and was 
proved for bundles as Theorem~A.5 in~\cite{Sanders_dirac:2010} based on a scaling limit structure of the Green operators in terms of Riesz distributions. 
Closely related statements for Feynman parametrices can be found in~\cite{IslamStrohmaier:2020}. 

Our main interest is in computing the polarisation sets of these distributions. The motivation for this
arises in quantum field theory~\cite{Fewster:2025a}, in the description of \emph{Hadamard states} -- generally agreed to be the class of physical states for linear field theories. For example, consider
the Klein--Gordon theory, with classical field equation $(\Box +m^2)\phi=0$. In the quantised theory, every sufficiently regular state determines a distributional two-point function $W\in\DD'(M\times M)$.
The Hadamard condition for states of the quantised theory was originally stated in terms of a Hadamard series expansion of the two-point function at nearby points~\cite{KayWald-PhysRep}, which is somewhat unwieldy when stated with full precision~\cite{KayWald-PhysRep,Moretti:2021}. A major breakthrough was the realisation by  Radzikowski~\cite{Radzikowski_ulocal1996} that a state is Hadamard if and only if its two-point function satisfies the wavefront set condition
\begin{equation}\label{eq:RcHad}
	\WF(W)=\Rc^\Had:=\{(x,k;x',k')\in \Rc: k\in\Nc^+\},
\end{equation}
where $\Nc^+\subset \Nc$ is the bundle of future-pointing nonzero null covectors. 
The Hadamard condition was subsequently described in similar terms for various other linear field theories of interest, notably those based on the Dirac, Proca, Maxwell and linearised Einstein equations. See e.g.,~\cite{Hollands:2001,Sanders_dirac:2010,Few&Pfen03,Hunt:2012thesis} for original sources and~\cite{BeniniDappiaggi:2015} for an exposition of some of the theories mentioned. Here one encounters a problem, because these field equations are not normally hyperbolic. For example, the Proca equation describes $1$-forms $A$ obeying
\begin{equation}
	(-\delta \dd + m^2)A = 0,
\end{equation}
where $\dd$ and $\delta$ are the exterior derivative and codifferential on $(M,g)$ and $m>0$ is the mass parameter.
As $-\delta\dd A=\nabla^\mu (\nabla_\mu A_\nu - \nabla_\nu A_\mu)$, it is clear that the Proca operator
is not normally hyperbolic. However its solution theory is closely related to that of the normally hyperbolic $1$-form Klein--Gordon operator $K^{(1)}=-\delta\dd+ \dd\delta+m^2$, and a short calculation shows that the operators
\begin{equation}\label{eq:EP_Proca}
	E_P^\pm = E_{K^{(1)}}^\pm \circ R
\end{equation}
provide advanced and retarded Green operators for $P=-\delta \dd + m^2$, where $R = 1-m^{-2}\dd\delta$. The definition of Hadamard states for the Proca field given in~\cite{Few&Pfen03} made use of this relationship: a state of the Proca field was said to be Hadamard if its two-point function obeys $W = H\circ (1\otimes R)$ and $H$ is a Hadamard-form bisolution for $K^{(1)}$ obeying $\WF(H) = \Rc^\Had$.

In an interesting recent paper~\cite{MorettiMurroVolpe:2023}, Moretti, Murro and Volpe (MMV) revisited the Proca theory with the aim of giving a more direct definition of the Hadamard condition that directly constrains $\WF(W)$. To do this, they need to compute 
$\WF(E_P)$. Standard bounds on the wavefront set, together with the known form of $\WF(E_{K^{(1)}})$ yield
\begin{equation}
	 \WF(E_P)\subset \Rc \subset \WF(E_P)\cup (\Char(1\otimes R)\cap \Rc),
\end{equation} 
so if the last term on the right-hand side is empty then one has $\WF(E_P)=\Rc$. This is what is claimed in 
part~(b) of the proof of Proposition~4.7 in~\cite{MorettiMurroVolpe:2023}, unfortunately on the basis of an incorrect 
assertion that $\Char(1\otimes R)=T^*M\times 0_{T^*M}$, where for any vector bundle $B$, $0_{B}$ denotes the corresponding zero section (and we drop the subscript where there is no ambiguity). In fact, the principal symbol 
of $R$ acts on $T_x^*M$ by $r(x,k) v  = -m^{-2}g^{-1}(k,v)k$, which vanishes on the annihilator of $k^\sharp$, so $\Char(R)=\dot{T}^*M$ and thus $\Char(1\otimes R)=\dot{T}^*(M\times M)$. Accordingly, 
the standard calculus of wavefront sets cannot be used to determine $\WF(E_P)$ and the MMV paper contains a gap at this point, affecting a number of their results (see~\cite{Fewster:2025a} for further discussion). A similar situation arises for the Dirac equation, whose Green operators are given in terms of the normally hyperbolic spinorial Klein--Gordon operator composed with another operator that, inconveniently, is characteristic on $\Nc$ (see e.g.,~\cite{Kratzert:2000} for discussion).

In this paper, we will establish general results that enable the gap in MMV to be closed -- with the 
result that $\WF(E_P)=\Rc$. The main tools are Dencker's polarisation set and his associated results on the propagation of polarisation~\cite{Dencker:1982}. We recall that the polarisation set $\WF_\pol(u)$ of a vector-valued distribution $u\in\DD'(B\otimes\Omega^{1/2})$ is a subset of the pullback bundle $\pi^*B$, where $\pi:T^*M\to M$ is the bundle projection, so that
\begin{equation}\label{eq:WFpol_def}
	\WF_\pol(u) = \bigcap_{Au\in \Gamma^\infty(\Omega^{1/2})} \{ (x,\xi;w)\in \pi^* B:\xi\neq 0,~w\in\ker \sigma(A)(x,\xi)	\},
\end{equation}	
where the intersection is taken over properly supported classical pseudodifferential operators $A\in\Psi_{\cl}^0(B\otimes\Omega^{1/2},\Omega^{1/2})$ that render $u$ smooth, and $\sigma(A)$ denotes the 
homogeneous principal symbol of $A$. Evidently $\WF_\pol(u)$ contains the zero section of $\pi^*B$ over $\dot{T}^*M:=T^*M\setminus 0$; the wavefront set $\WF(u)$ is given by the set of $(x,k)\in\dot{T}^*M$ over which the fibre of $\WF_\pol(u)$ is nonzero. Further properties of polarisation sets will be described in Section~\ref{sec:polset}.  

To describe our main results, let $P$ be a normally hyperbolic operator on $\Gamma^\infty(B\otimes\Omega^{1/2})$. Let $\nabla^B$ be the Weitzenb\"ock connection on $B$ induced by the
normally hyperbolic operator $(-g)^{-1/4}P(-g)^{1/4}$, and let $\nabla^{B^*}$ be the dual connection on $B$. Next, let $\Pi_{x,k}^{x',k'}\in  \Lin(B_{x'},B_{x})$ be the operator of parallel transport
with respect to $\nabla^B$ along the (unique) null geodesic segment connecting $x'$ to $x$ that witnesses to the relation $(x,k)\sim (x',k')$. If $(x,k)\not\sim (x',k')$, then we set $\Pi_{x,k}^{x',k'}=0$. We can now define the sets 
\begin{equation}\label{eq:Rcpolsets}
	\Rc_\pol^\# =\{(x,k;x',-k';w)\in (\pi\times\pi)^* (B\boxtimes B^*): (x,k;x',-k')\in\Rc^\#,~
	w\in \CC \Pi_{x,k}^{x',k'}
	\},  
\end{equation}
where $\#$ stands for $+$, or $-$ or no symbol, and the sets $\Rc^\#$ were defined in~\eqref{eq:Rcsets}. We have also identified $\Lin(B_x',B_x)$ with $B_x\otimes B_{x'}^*$. At first sight, there appears to be 
an asymmetry in~\eqref{eq:Rcpolsets}, but note that $(x,k)\sim(x',k')$ if and only if $(x,-k)\sim(x',-k')$ and $\Pi_{x,k}^{x',k'}=\Pi_{x,-k}^{x',-k'}$. Our main result can now be stated.
\begin{thm}\label{thm:WFpolEPpmandEP}
	Let $B$ be a finite-rank complex vector bundle over a globally hyperbolic spacetime $(M,g)$.
	Let $P$ be a normally hyperbolic operator on $\Gamma^\infty(B\otimes\Omega^{1/2})$, and $\Rc_\pol^\pm$, $\Rc_\pol$ be as defined in~\eqref{eq:Rcpolsets}. Then:
	\begin{enumerate}[a)]
	\item The polarisation sets of the kernel distributions for the Green operators $E^\pm_P$ of $P$ satisfy
	\begin{equation}\label{eq:WFpolEPpm}
		\WF_\pol(E_P^\pm) = \Rc^\pm_\pol \cup \WF_\pol(\id),
	\end{equation}
	where the polarisation set of the kernel distribution of the identity operator on $\Gamma_0^\infty(B\otimes\Omega^{1/2})$ is
	\begin{equation}\label{eq:WFpol_id}
		\WF_\pol(\id)= \{(x,k;x,-k;w)\in (\pi\times \pi)^*(B\boxtimes B^*): k\in \dot{T}^*_xM,~w\in\CC\delta_x\}\cup 0,
	\end{equation}
	in which $\delta_x$ is the image of $\id_{B_x}$ under the identification of $\Lin(B_x)$ with $B_x\otimes B_x^*$.
	
	\item The polarisation set of the kernel distribution of $E_P=E_P^--E_P^+$ is
	\begin{equation}
		\WF_\pol(E_P) = \Rc_\pol\cup 0.
	\end{equation}
	\item Consequently, the wavefront sets of $E_P^\pm$ and $E_P$ are 
	\begin{equation}
		\WF(E_P^\pm)=\Rc^\pm\cup \WF(\id), \qquad   \WF(E_P)=\Rc.
	\end{equation}
	\end{enumerate}
\end{thm} 
Part~(c) follows trivially from~(a) and~(b), and provides an alternative proof of~\eqref{eq:WFEPpm} and~\eqref{eq:WFEP} that avoids scaling limits or an expansion into Riesz distributions. Special cases of Theorem~\ref{thm:WFpolEPpmandEP} have been studied before~\cite{Kratzert:2000} but the result seems to be new in this generality. 

A simple consequence of Theorem~\ref{thm:WFpolEPpmandEP} is
\begin{cor}\label{cor:WFpol2}  Under the hypotheses of Theorem~\ref{thm:WFpolEPpmandEP},  
	suppose that $Q\in\Psi_\cl^{m_Q}(B\otimes \Omega^{1/2},\hat{B}\otimes\Omega^{1/2})$ and
	$R\in\Psi_\cl^{m_R}(\tilde{B}\otimes\Omega^{1/2},B\otimes\Omega^{1/2})$ have principal symbols $q$ and $r$, where $\hat{B}\to M$ and
	$\tilde{B}\to M$ are finite-rank complex vector bundles. 
	
	If for all $(x,k;x',-k')\in \Rc$ one has
	\begin{equation}
		q(x,k)\circ \Pi_{x,k}^{x',k'}\circ r(x',k')\neq 0
	\end{equation}
	then $\WF(QE_PR)=\Rc$. This holds in particular if
	$Q=\id$ and $r$ is nonvanishing on $\Nc$, or if $R=\id$ and $q$ is nonvanishing on $\Nc$.
\end{cor}
Using Corollary~\ref{cor:WFpol2} we can close the gap in MMV quite easily, because~\eqref{eq:EP_Proca} corresponds to the situation where $Q=\id$ and $r$ is nonvanishing, albeit everywhere characteristic. We will go into more detail about this in section~\ref{sec:Proca}, where we will also compute the polarisation set of the Proca advanced-minus-retarded operator.

Corollary~\ref{cor:WFpol2} is intended as a useful general tool to facilitate the use of Theorem~\ref{thm:WFpolEPpmandEP}, particularly when applied to quantum field theory in curved spacetimes.
Further applications appear in~\cite{Fewster:2025a,FewsterKlein:2025}. In particular, Ref.~\cite{Fewster:2025a} introduces a generalised Hadamard condition for states of QFTs based on a class of real, formally hermitian, Green-hyperbolic operators (\rfhgho's) satisfying a decomposability condition on the wavefront set of the advanced-minus-retarded operator $E_P$, namely that
\begin{equation}
	\WF(E_P) \subset (\Vc^+\times \Vc^-)\cup (\Vc^-\times\Vc^+),
\end{equation}
where $\Vc^\pm$ are nonintersecting opposite conic sets ($\Vc^-=-\Vc^+$) that are relatively closed in $\dot{T}^*M$, whereupon a state is said to be $\Vc^+$-Hadamard if its two-point function $W$ satisfies $\WF(W)\subset \Vc^+\times\Vc^-$.	The cones $\Vc^\pm$ need not be related to the spacetime metric. Corollary~\ref{cor:WFpol2} is used in cases (including the Proca equation) where the advanced-minus-retarded operator of a \rfhgho\ $T$ takes the form $E_T= QE_PR$, where $P$ is a normally hyperbolic operator. The conclusion is that $T$ is then decomposable with respect to the usual null cones $\Nc^\pm$ and that the $\Nc^+$-Hadamard condition implies~\eqref{eq:RcHad} in such cases --- see Corollary 5.6 of~\cite{Fewster:2025a}. These results also feed into Section~7 of~\cite{Fewster:2025a}, where the existence of Hadamard states for the neutral Proca field on arbitrary globally hyperbolic spacetimes is established and the $\Nc^+$-Hadamard condition is shown to be equivalent to the MMV definition and the definition given originally in~\cite{Few&Pfen03}. Meanwhile, in~\cite{FewsterKlein:2025} results from~\cite{Fewster:2025a} are used to show that Hadamard states exist for the charged Proca field in an external electromagnetic field and for a variety of other coupled Proca theories. The solution theory of these systems turns out to be quite involved, but can be related to that of normally hyperbolic operators so that Corollary~\ref{cor:WFpol2} applies. We expect that other applications will follow.

The paper is structured as follows. After some preliminaries in section~\ref{sec:psdo}, the polarisation set is described in section~\ref{sec:polset} along with its basic properties. In particular, Dencker's results on the propagation of polarisation are set out in section~\ref{sec:prop_pol}. Our main results, Theorem~\ref{thm:WFpolEPpmandEP} and Corollary~\ref{cor:WFpol2}, are proved in section~\ref{sec:polGreen}.
These results are applied to the specific example of the neutral Proca field in section~\ref{sec:Proca}, where they are used to close the gap in the MMV paper~\cite{MorettiMurroVolpe:2023}. In fact we go beyond what is necessary for that task and compute the full polarisation set of the advanced-minus-retarded operator for the Proca field in any globally hyperbolic spacetime. This has the interesting, and perhaps inconvenient, feature that it is influenced by the constraint that distinguishes solutions of the Proca equation among Klein--Gordon solutions, rather than by the propagating physical degrees of freedom. We conclude with a summary and outlook in section~\ref{sec:Conc}. To make the presentation in sections~\ref{sec:polset} and~\ref{sec:polGreen} accessible to a wider readership, some of the basic definitions of pseudodifferential operator theory are summarised in appendix~\ref{appx:psido}.

To conclude this introduction, we mention some further literature on Hadamard states, particularly in relation to pseudodifferential techniques and polarisation sets. A wider selection of references is given in~\cite{Fewster:2025a}, but it is relevant to mention that pseudodifferential methods have played a role in the construction of Hadamard and adiabatic states since the works of Junker~\cite{Junker:1996,Junker:1996_erratum} and Junker and Schrohe~\cite{JunkerSchrohe}, and subsequently by G\'erard, Wrochna and coauthors who have given constructions for Hadamard states of the scalar field~\cite{GerardWrochna:2014,GerardOulghaziWrochna:2017} and also linear theories with gauge freedom including the linearised Yang--Mills~\cite{GerardWrochna:2015} and linearised gravity~\cite{Gerard:2023}; see also~\cite{Gerard:2019} for an exposition. Microlocal formulations of Hadamard states for vector-valued fields were first studied in generality by Sahlmann and Verch~\cite{SahlmannVerch:2000RMP}, including the Dirac field, while the works on the Dirac field by Kratzert~\cite{Kratzert:2000} and Hollands~\cite{Hollands:2001} also include discussion of the polarisation set of the two-point function. However it was not clear how much was gained by using the polarisation set, and a consensus formed -- shared by this author -- that the wavefront set was sufficient for the description of Hadamard states. So, for example, Sanders' comprehensive treatment of the Dirac field~\cite{Sanders_dirac:2010} only refers to the polarisation set as a means to prove propagation of singularities results and G\'erard's monograph~\cite{Gerard:2019} does not mention it at all. The recent discussion of Hadamard states for Proca fields~\cite{MorettiMurroVolpe:2023,Fewster:2025a,FewsterKlein:2025} has shown that the consensus view was not entirely correct and that there are some occasions where the polarisation set provides valuable finer detail that is encapsulated in the results presented here.

\section{Background on distributions and pseudodifferential operators}\label{sec:psdo}

\paragraph{General conventions} The symbol $\subset$ does not exclude equality. When symbols are stacked vertically, e.g., $E^\pm$, the alternatives are ordered from top to bottom. 
Duality pairings $V^*\times V\to \CC$ are indicated by double
angle-brackets, $\dlangle \cdot,\cdot\drangle_{V}$; the subscript is dropped where there is no ambiguity.
Lorentzian metrics have mostly minus signature.

\paragraph{Manifolds and bundles} Let $X$ be a smooth $n$-manifold. If $E$ is a finite-rank complex vector bundle over $X$ then $\Gamma^\infty(E)$ denotes the space of smooth sections of $E$ and $\Gamma^\infty_0(E)$ those of compact support. The bundle of weight-$\alpha$ densities over $X$ is denoted
$\Omega^\alpha_X$, suppressing the subscript if there is no ambiguity about the base manifold. The notation $\dot{T}^*X$ denotes the cotangent bundle of $X$ with its zero section excised. For bundles 
$E$ and $F$, $E\boxtimes F$ is the external tensor product bundle over the Cartesian product of the base spaces, while, if applicable, $E\otimes F$ is the tensor product bundle over the common base space of $E$ and $F$.

\paragraph{Distributions and kernels} We regard distributions as generalised functions, so $\DD'(X)=\Gamma_0^\infty(\Omega_X)'$.
Similarly, for a bundle $E$ over $X$, $\DD'(E)= \Gamma_0^\infty(E^*\otimes\Omega_X)'$ 
 with the weak topology, and there is an inclusion $\iota_X:\Gamma^\infty(E)\to\DD'(E)$  
\begin{equation}
	(\iota_X e)(u) = \int_X \dlangle u,e\drangle_{E}, \qquad e\in \Gamma^\infty(E),~u\in \Gamma_0^\infty(E^*\otimes\Omega_X).
\end{equation}
Here, we have extended the notation for duality pairing to densitised sections, so that it takes values in densities of the total weight of the arguments. 
Half-densitised sections have the convenient feature that $\DD'(E\otimes\Omega_X^{1/2})=\Gamma_0^\infty(E^*\otimes\Omega^{1/2}_X)'$.

Next, suppose $E$ and $F$ are bundles over manifolds $X$ and $Y$, and
$T:\Gamma_0^\infty(E\otimes \Omega^{1/2}_X)\to \Gamma^\infty(F\otimes \Omega^{1/2}_Y)$ is a continuous linear map. Then there is a kernel
distribution $T^\knl\in \DD'((F\boxtimes E^*)\otimes (\Omega_{Y\times X}^{1/2}))$ defined by
\begin{equation}
	T^\knl(v\otimes e)= \int \dlangle v, Te\drangle_{F}, \qquad
	e\in \Gamma_0^\infty(E\otimes \Omega^{1/2}),~v\in \Gamma_0^\infty(F^*\otimes \Omega^{1/2})
\end{equation}
Typically we will write the $\knl$ superscript only for emphasis.

\paragraph{Pseudodifferential operators}
Let $X$ be a smooth $n$-manifold and let $\pi:T^*X\to X$ be the projection map. For any $m\in\RR$, $\Psi_{\cl}^m(X)$ will denote the space of classical (i.e., properly supported, $1$-step polyhomogeneous) pseudodifferential operators on $X$. For convenience the definition of  $\Psi_{\cl}^m(X)$ and related concepts, such as the symbol classes $S^m(X)$, are recalled in Appendix~\ref{appx:psido} following~\cite{Hormander3}, while the following summary serves to fix notation.

Each $A\in \Psi_{\cl}^m(X)$ is in particular a continuous linear map
$A:C_0^\infty(X)\to C_0^\infty(X)$ that extends continuously to $A:\EE'(X)\to\EE'(X)$. There is an associated principal symbol $\sigma(A)$, often written $a$, which belongs to the symbol class $S^m(X)\subset C^\infty(T^*X)$ and is homogeneous of degree $m$ in the covector argument outside a neighbourhood of the zero section, where it is uniquely determined by $A$; conversely $\sigma(A)$ uniquely determines $A$ up to addition of elements in $\Psi_\cl^{m-1}(X)$. A basic example is that, for any vector field $V\in \Gamma^\infty(TX)$, the symbol given by $a(x,\xi)=\ii\dlangle \xi,V_x\drangle$ corresponds to the directional derivative operator $\nabla_V$ modulo operators in $\Psi_\cl^0(X)$.  
Every chart on $X$ determines a chart representative of $A\in \Psi_\cl^m(X)$, which is specified up to smoothing operators by its full symbol.

If $E$ and $F$ are finite-rank complex vector bundles over $X$, 
then $\Psi_\cl^m(E,F)$ is the
space of continuous linear maps $A:\Gamma_0^\infty(E)\to \Gamma_0^\infty(F)$ so that, for each $e\in \Gamma^\infty(E)$, $v\in\Gamma^\infty(F^*)$, one has
$A^v_{\phantom{v}e}\in\Psi_\cl^m(X)$, where $A^v_{\phantom{v}e}f=\dlangle v,A (fe)\drangle$. We write
$\Psi_\cl^m(E)$ for $\Psi_\cl^m(E,E)$. Each $A\in \Psi_\cl^m(E,F)$ has a principal symbol
$a\in S^m(\pi^*\Hom(E,F))$ so that
\begin{equation}
	\sigma(A^v_{\phantom{v}e})(x,\xi) = \dlangle v_x,a(x,\xi) e_x\drangle, 
\end{equation}
for all $e\in\Gamma^\infty(E)$, $v\in\Gamma^\infty(F^*)$, and $(x,\xi)\in T^*X$. We define 
\begin{equation}
	\ker a = \{(x,\xi;w)\in \pi^*E: \xi\neq 0,~a(x,\xi)w=0\},
\end{equation}
and describe $A$ or $a$ as \emph{characteristic} (resp., \emph{noncharacteristic})at $(x,\xi)\in\dot{T}^*X$ if 
$\ker a|_{(x,\xi)}\subset E_x$ is nontrivial (resp., trivial), also writing
\begin{equation}
	\Char A = \{(x,\xi)\in\dot{T}^*X: \ker a(x,\xi)\neq 0\}.
\end{equation}  
If $E$ and $F$ have equal rank, $a$ is noncharacteristic at $(x,\xi)$ if and only if there is a symbol in $S^{-m}(\pi^*\Hom(F,E))$ that is inverse to 
$a$ in a conic neighbourhood of $(x,\xi)$. 

\paragraph{Subprincipal and refined principal symbols} Any operator $A\in\Psi^m_\cl(\Omega^{1/2})$ has a \emph{refined principal symbol} $a^\textnormal{r}\in S^m(X)/S^{m-2}(X)$ defined as follows (see \cite[18.1.33]{Hormander3}). Let $(X_\kappa,\kappa)$ be a chart in $X$ and let the corresponding chart representative of $A$ have full symbol $a_\kappa\in 
S^m(\kappa(X_\kappa)\times\RR^n)$. One defines $a^\textnormal{r}_\kappa\in C^\infty(T^*\RR^n|_{\kappa(X_\kappa)})$ by
\begin{equation}\label{eq:refinedprincipal}
	a_\kappa^\textnormal{r}(x,\xi) = a_\kappa(x,\xi) + \frac{\ii}{2}\frac{\partial^2 a_\kappa}{\partial x^\mu\partial \xi_\mu}(x,\xi) .
\end{equation}
Then $a^\textnormal{r}\in S^{m}(T^*X)/S^{m-2}(T^*X)$ is uniquely defined by requiring $a^\textnormal{r}\circ  \kappa^* -
a_\kappa^\textnormal{r}\in S^{m-2}(\kappa(X_\kappa)\times\RR^n)$ for all local coordinate systems. (Here, we identify $\kappa(X_\kappa)\times\RR^n$ with $T^*\kappa(X_\kappa)$ for the purpose of applying the pullback $\kappa^*:T^*\kappa(X_\kappa)\to T^*X$.) Furthermore, $A$ has a \emph{subprincipal symbol} $a^\sub\in S^{m-1}(X)$ that, outside a neighbourhood of the zero section, is homogeneous of degree $m-1$ in the covector argument and obeys $a^\textnormal{r}=a+a^\sub \pmod {S^{m-2}(X)}$.  

The subprincipal symbol generalises to operators $A\in\Psi^m_\cl(E\otimes\Omega^{1/2})$ with scalar principal symbol, $a = b\id_E$. Consider sections $e_r\in \Gamma^\infty(E)$ and
$v_r\in \Gamma^\infty(E^*)$ ($1\le r\le \textnormal{rk}E$) that provide a local framing for $E$ and a dual framing of $E^*$.
Then one has an $r\times r$-matrix $(A^{s}_{\phantom{s}r})_{r,s=1}^{\textnormal{rk} E}$ with entries $A^{s}_{\phantom{s}r}:=A^{v_s}_{\phantom{v_s}e_r}\in \Psi^m_\cl(\Omega^{1/2})$, each of which has its own refined principal symbol. The matrix of refined principal symbols $a^r$ does not transform as a section of $\pi^*\Hom(E)$, but there is a related invariantly defined object obtained as follows. Let $\nabla^E$ be any connection on $E$, and define connection $1$-forms with respect to the framing $e_r$ by $(\Gamma^E_V)^{s}_{\phantom{s}r} =\dlangle v_s, \nabla^E_V e_r\drangle$ for any $V\in \Gamma^\infty(TM)$. The pullback connection $\nabla^{\pi^*E}$ on $\pi^*E$ is defined so that $\nabla_W^{\pi^*E}\pi^*s |_{(x,k)}= \nabla_{\pi_* W}^E s|_x$ for $(x,k)\in T^*M$, $s\in \Gamma^\infty(E)$ and $W\in T_{(x,k)}T^*M$. Its connection $1$-forms obey $\Gamma^{\pi^*E}_W = \Gamma^E_{\pi_*W}\circ \pi$ relative to the framing $\pi^* e_r$. 
The connection $1$-forms and refined principal symbols have a related transformation law under a change of framing. Namely, 
if $X_b$ is the Hamiltonian vector field
\begin{equation}
	X_b = \frac{\partial b_\kappa}{\partial \xi_\mu}\frac{\partial}{\partial x^\mu} - \frac{\partial b_\kappa}{\partial \xi_\mu}\frac{\partial}{\partial \xi^\mu}
\end{equation}
determined by $b$, then the matrix $a^r+\ii\Gamma^{\pi^*E}_{X_b}$ 
transforms, modulo $S^{m-2}(\pi^*\Hom(E))$, as a section of  
$\pi^*\Hom(E)$ (see Appendix~\ref{appx:psido} for more detail on this point). In this way, the matrix of subprincipal symbols
$a^\sub = a - a^\textnormal{r}$ can be regarded as a partial connection along the Hamiltonian flow of $b$.

\paragraph{The formal dual}
Each $A\in\Psi^m_\cl(E\otimes\Omega^{1/2},F\otimes\Omega^{1/2})$ has a formal dual $\sadj{A}\in 
\Psi^m_\cl(F^*\otimes\Omega^{1/2},E^*\otimes\Omega^{1/2})$, so that $(\sadj{A}u)(f)=u(Af)$
for $u\in\DD'(F^*\otimes\Omega^{1/2})$, $f\in\Gamma_0^\infty(E\otimes\Omega^{1/2})$. Its principal symbol is
\begin{equation}
	(\sadj{a})(x,\xi) = a(x,-\xi)^*,
\end{equation}
where the star on the right-hand side denotes the dual map. This is easily obtained  as a modification of Theorem~18.1.7 (cf.\ also Theorem 18.1.34${}'$) in~\cite{Hormander3}.

Suppose $E$ and $F$ are finite-rank complex vector bundles over smooth manifolds $X$ and $Y$. If $T:\Gamma_0^\infty(E\otimes \Omega^{1/2}_X)\to \Gamma^\infty(F\otimes \Omega^{1/2}_Y)$ is a continuous linear map and $R\in\Psi^m_\cl(E\otimes\Omega_X^{1/2})$, 
$Q\in\Psi^{m'}_\cl(F\otimes\Omega_Y^{1/2})$, then one has the identity
\begin{equation}
	(QTR)^\knl = (Q\otimes\sadj{R})T^\knl
\end{equation}
 for the kernel distributions.

\section{The polarisation set}\label{sec:polset}

\subsection{Definition and basic properties} 

We take the calculus of wavefront sets for scalar or densitised distributions for granted~\cite{Hormander1,Hoer_FIOi:1971}. For 
$u\in\DD'(E\otimes\Omega^{1/2})$, one defines
\begin{equation}\label{eq:WFbndldef}
	\WF(u) = \bigcup_{s\in\Gamma^\infty(E^*)} \WF(f\mapsto u(f s)),
\end{equation}
which can be restricted to a union over sections forming a local frame for $E^*$ near $x$, for the purposes of computing $\WF(u)$ at $x$. It is easily checked that various standard results on scalar distributions generalise immediately. In particular, one has 
\begin{equation}\label{eq:WFPu}
	\WF(Pu)\subset \WF(u)
\end{equation}
for all pseudodifferential operators $P$. Similarly, a sufficient condition for a continuous linear map $T:\Gamma_0^\infty(E\otimes \Omega^{1/2}_X)\to \Gamma_0^\infty(F\otimes \Omega^{1/2}_Y)$ to act on $u\in \DD'(E\otimes\Omega^{1/2}_X)$ is that 
$\WF'(T^\knl)\cap (0\times\WF(u))=\emptyset$, whereupon
\begin{equation}
	\WF(Tu)\subset \WF'(T^\knl)\bullet \WF(u), 
\end{equation} 
where
\begin{align}
	\WF'(T^\knl)\bullet \WF(u)&:=\{(x,\xi)\in \dot{T}^*X: \exists~(y,\eta)\in\WF(u),~\textnormal{s.t.}~
	(x,\xi;y,-\eta)\in\WF(T^\knl) \} \nonumber\\
	&= \pr_{\dot{T}^*X} \WF'(T^\knl)\cap (\dot{T}^*X\times\WF(u)).
\end{align}
Scalar distributions also obey $\WF(u)\subset \WF(Pu)\cup (\Char P\cap\WF(u))$ but this result does not follow easily from~\eqref{eq:WFbndldef}. It can be derived instead using the polarisation set, to which we now turn.

Let $\pi:T^*X\to X$ be the bundle projection. As already mentioned, the polarisation set~\cite{Dencker:1982} $\WF_\pol(u)$ of  $u\in \DD'(E\otimes\Omega^{1/2})$ is the subset of the pullback bundle $\pi^*E$ over $T^*X$ given by 
\begin{equation}
	\WF_\pol(u) = \bigcap_{Au\in \Gamma^\infty(\Omega^{1/2})} \{ (x,\xi;w)\in \pi^* E:\xi\neq 0,~w\in\ker \sigma(A)(x,\xi) 
	\},
\end{equation}	
taking the intersection over $A\in\Psi_{\cl}^0(E\otimes\Omega^{1/2},\Omega^{1/2})$ so that $Au$ is smooth. In fact, Dencker introduced polarisation sets for 
distributions valued in $\CC^N$ but (as mentioned in~\cite{Dencker:1982}) the generalisation to bundles is straightforward.
It is convenient to define the polarisation set as a subset of $\pi^*E$ rather than $\pi^*E\otimes\Omega^{1/2}$, because the propagation of polarisation result for $P\in\Psi_\cl^m(E\otimes\Omega^{1/2})$ of real principal type turns out to be described using a connection on $\pi^*E$ pulled back from $E$. Note that
$\dot{T}^*X\times 0\subset \WF_\pol(u)$ for every $u\in \DD'(E\otimes\Omega^{1/2})$.

The relationship between the polarisation set and the wavefront set is 
\begin{equation}\label{eq:WFfromWFpol}
	\WF(u) = \tilde{\pi}(\WF_\pol(u)\setminus 0)=
	\{(x,\xi)\in \dot{T}^*X: \exists  w\in E_x\setminus \{0\}~\textnormal{s.t.}~(x,\xi;w)\in\WF_\pol(u)\},
\end{equation}
as shown in Proposition~2.5 of~\cite{Dencker:1982},
so $\WF_\pol(u)\subset \tilde{\pi}^{-1}(\WF(u))\cup 0$. Here, $\tilde{\pi}:\pi^*E\to T^*X$ is the induced bundle projection and we have abbreviated $\dot{T}^*X\times 0$ by $0$.

The following is a basic result on polarisation sets.
\begin{lemma}\label{lem:WFpol0}
	Let $u\in \DD'(E\otimes\Omega^{1/2})$. If $A\in \Psi_\cl^m(E\otimes\Omega^{1/2},F\otimes\Omega^{1/2})$ with homogeneous principal symbol $a\in S^m(\pi^*\Hom(E,F))$, then 
	\begin{equation}\label{eq:aWFpolusubsetWFpolAu}
		a(\WF_\pol(u))\subset \WF_\pol(Au),
	\end{equation}
	where $a$ acts fibrewise, $a(x,\xi;w) = (x,\xi;a(x,\xi)w)$. 
	In the special case where $E$ and $F$ have equal rank and $A$ is noncharacteristic at $(x,\xi)\in \dot{T}^*X$, the inclusion in the above equation can be replaced by equality over a conical neighbourhood of $(x,\xi)$. 
	Consequently (returning to the general case),
	if for some $(x,\xi)\in T^*X$ one has 
	\begin{equation}\label{eq:nontrivfibre}
		a(x,\xi)(\WF_\pol(u)|_{(x,\xi)})\neq 0,
	\end{equation} 
	then $(x,\xi)\in \WF(Au)$.  If~\eqref{eq:nontrivfibre} holds for all $(x,\xi)\in\WF(u)$, then $\WF(Au)=\WF(u)$.	
\end{lemma}
\begin{proof}
	The first two parts rewrite Proposition~2.7 and Corollary 2.8 of~\cite{Dencker:1982} in the bundle context.	
	For the consequence, by hypothesis, there exists $(x,\xi;w)\in \WF_\pol(u)$ so that $a(x,\xi)w\neq 0$. Eq.~\eqref{eq:aWFpolusubsetWFpolAu} implies 
	$\WF_\pol(Au)|_{(x,\xi)}\neq 0$, whereupon~\eqref{eq:WFfromWFpol} implies that $(x,\xi)\in \WF(Au)$. The additional hypothesis now implies $\WF(u)\subset \WF(Au)$, and $\WF(Au)\subset \WF(u)$ holds in general.
\end{proof}
Lemma~\ref{lem:WFpol0} implies that $\WF_\pol(u)\subset a^{-1}(\WF_\pol(Au))\subset a^{-1}(\WF_\pol(Au)\setminus 0)\cup \ker a$, where
\begin{equation}
	\ker a = \{(x,\xi;w)\in\pi^* E: w\in\ker a(x,\xi)\},
\end{equation}
giving
\begin{equation}\label{eq:WFpolgenbds}
	\WF_\pol(u) \subset  a^{-1}(\WF_\pol(Au)\setminus 0)\cup (\WF_\pol(u)\cap \ker a)
	\subset a^{-1}(\WF_\pol(Au)\setminus 0)\cup  \ker a,
\end{equation}
and hence one obtains the bound 
\begin{equation}
	\WF(u)\subset \WF(Au)\cup (\WF(u)\cap \Char A)\subset \WF(Au)\cup\Char A,
\end{equation}
which is well-known in its scalar version~\cite{Hormander1}.
Note that $\WF_\pol(Au)$ cannot be bounded above in terms of $\WF_\pol(u)$, because the polarisation set only tracks the strongest singularity (see below). However, the standard result $\WF(Au)\subset \WF(u)$ implies  
\begin{equation}
	\WF_\pol(Au)\subset \pi^*E|_{\WF(u)}\cup 0.
\end{equation} 

\paragraph{Examples}
1. We compute $\WF_\pol(\id)$, where $\id$ is the kernel distribution of the identity operator on $E\otimes\Omega^{1/2}$. For any $A\in\Psi_\cl^m(E\otimes\Omega^{1/2})$, the operator identity $A\circ \id=\id\circ A$ translates to the kernel identity $(A\otimes 1)\id^\knl=(1\otimes \sadj{A})\id^\knl$. Thus if $w\in\WF_\pol(\id)|_{(x,\xi;x',-\xi')}$ then $(a(x,\xi)\otimes 1)w=
(1\otimes a(x',\xi')^*)w$ for all $a$, which implies that $w=0$ unless
$x=x'$, $\xi=\xi'$, in which case $w\in\CC\delta_x$, the element of $E_x\otimes E_x^*$ identified with $\id_{E_x}$. Thus
\begin{equation}
		\WF_\pol(\id)\subset \{(x,\xi;x,-\xi;w)\in (\pi\times \pi)^*(E\boxtimes E^*): \xi\in \dot{T}^*_xX,~w\in\CC\delta_x\}\cup 0.
\end{equation}
Conversely, the bound $\WF(u)=\WF(\id\, u)\subset \WF'(\id^\knl)\bullet \WF(u)$ ($u\in\DD'(E\otimes\Omega^{1/2})$) shows that $\WF(\id^\knl)$ contains the twisted diagonal $\{(x,\xi;x,-\xi):(x,\xi)\in\dot{T}^*X\}$ and so $\WF_\pol(\id)$ must be nontrivial over all such points, giving the result stated in~\eqref{eq:WFpol_id}.  

2. As a (slightly cautionary) example, consider any $u\in\DD'(E\otimes\Omega^{1/2})$ and
$A\in\Psi^m(E\otimes\Omega^{1/2},F\otimes\Omega^{1/2})$ with $m>0$. Let
$v=u\oplus Au\in\DD'((E\oplus F)\otimes\Omega^{1/2})$. Then 
$\pr_E v=u$ and $\pr_{F}=Au$, where $\pr_{E}$ and $\pr_{F}$ are the 
projection maps from $E\oplus F$ to $E$ and $F$ respectively. Applying
Lemma~\ref{lem:WFpol0}, one has
\begin{align}
	 \pr_E(\WF_\pol(v))&\subset \WF_\pol(\pr_E v)=\WF_\pol(u), \nonumber\\
	 \pr_{F}(\WF_\pol(v))&\subset \WF_\pol(\pr_{F} v)=\WF_\pol(Au).
\end{align}
However, one also has $\begin{pmatrix}	A & -1 \end{pmatrix} v = 0$,
in which $\begin{pmatrix}	A & -1 \end{pmatrix}\in \Psi^m_\cl((E\oplus F)\otimes\Omega^{1/2};F\otimes\Omega^{1/2})$ has symbol $\begin{pmatrix} a & 0 \end{pmatrix}$ (because $m>0$), giving the further constraint 
$a\circ \pr_E (\WF_\pol(v))=0$. Overall, this gives the upper bound
\begin{equation}
 \WF_\pol(u\oplus Au) \subset (\WF_\pol(u)\cap\ker a)\oplus \WF_\pol(Au).
\end{equation}
Thus some of the singularities in $\WF_\pol(u)\subset\pi^*E$ might not 
appear in $\WF_\pol(u\oplus Au)$, because they appear more strongly in $\pi^*F$ due to the action of $A$. 

The second example shows that polarisation sets do not have as convenient a calculus as that of wavefront sets. Even addition requires care.
\begin{lemma}\label{lem:WFpol_addition}
	Suppose that $u,v\in\DD'(E\otimes\Omega^{1/2})$. Then
	\begin{equation}
		\WF_\pol(u+v) =\WF_\pol(u)+ \WF_\pol(v),
	\end{equation}
	holds fibrewise over the complement of $\WF(u)\cap \WF(v)$.
\end{lemma}
\begin{proof}
	Let $(x,\xi)$ belong to at most one of $\WF(u)$ or $\WF(v)$, say $\WF(u)$, so $\WF_\pol(v)|_{(x,\xi)}=0$.
	Then $A\in\Psi^0_\cl(E\otimes\Omega^{1/2};\Omega^{1/2})$ makes $A(u+v)$ smooth at $(x,\xi)$ if and only if $Au$ is smooth at $(x,\xi)$, so $\WF_\pol(u+v)|_{(x,\xi)}=\WF_\pol(u)|_{(x,\xi)} = 
	\WF_\pol(u)|_{(x,\xi)}+ \WF_\pol(v)|_{(x,\xi)}$ as required.
\end{proof}
Our previous example shows that one cannot expect a result of this type above $\WF(u)\cap \WF(v)$. For we have seen that $\WF_\pol(u\oplus Au)$ can be 
strictly smaller than $\WF_\pol(u\oplus 0)+\WF_\pol(0\oplus Au)$, while, by the same token,
$\WF_\pol(u\oplus 0)$ can contain elements outside $\WF_\pol(u\oplus Au)+\WF_\pol(0\oplus (-Au))$.

\begin{lemma} \label{lem:WFpol1}
	Suppose that $u\in\DD'(\Omega^{1/2})$ and $s\in\Gamma^\infty(E)$.
	Then $su\in \DD'(E\otimes\Omega^{1/2})$ and 
	\begin{equation}\label{eq:WFpol1}
		\WF_\pol(su)  = \{(x,\xi;w)\in\pi^* E: (x,\xi)\in\WF(su),~w\in\CC s_x\}\cup 0 .
	\end{equation}
\end{lemma}
\begin{proof}
	Multiplication by $s$ defines $S\in\Psi_{\cl}^0(\Omega^{1/2}; E\otimes\Omega^{1/2})$ with symbol $(x,\xi)\mapsto s_x\in\Hom(\CC,E_x)$. 
	Now $\WF_\pol(u)\setminus 0 = \WF(u)\times \CC$, so Lemma~\ref{lem:WFpol0} shows that 
	\begin{equation}
		\{(x,\xi;w)\in\pi^* E: (x,\xi)\in\WF(u),~w\in\CC s_x\}\cup  0 \subset \WF_\pol(su).
	\end{equation} 
	Taking intersections with $\tilde{\pi}^{-1}(\WF(su))\cup  0$,
	one sees that the right-hand side of~\eqref{eq:WFpol1} is contained in the left-hand side.
	Conversely, we need only show that $\WF_\pol(su)\setminus 0\subset \CC \pi^*s$, because it is 
	automatic that $\WF_\pol(su)\subset \pi^*E|_{\WF(su)}$. 
	Consider the operator $S^\wedge\in \Psi_{\cl}^0(E\otimes\Omega^{1/2}; (E\otimes E)\otimes\Omega^{1/2})$
	acting on $v\in\DD'(E\otimes\Omega^{1/2})$ by $S^\wedge v= s\wedge v$, where $(s\wedge v)(f)= v(s\cdot f-f\cdot s)$ for $f\in\Gamma_0^\infty((E\otimes E)^*\otimes\Omega^{1/2})$ and 
	$(s\cdot f)^B=s_Af^{AB}$, $(f\cdot s)^A= f^{AB}s_B$ in an obvious index notation. Evidently,
	$s\wedge (su)=0$, so the polarisation set of $su$ at $(x,\xi)$ lies in the kernel of $w\mapsto s\wedge w$, i.e., 
	$\WF_\pol(su)\setminus 0\subset \CC \pi^*s$ as required. 
\end{proof}

\paragraph{Bundle pullbacks}
If $\psi:X\to Y$ is an embedding and $F$ is a bundle over $Y$, then one has the usual pullback
$\psi^*:\Gamma^\infty(F)\to\Gamma^\infty(\psi^*F)$, $(\psi^*s)(x)=s(\psi(x))$. As densities also pull back, we can extend this to $\psi^*:\Gamma^\infty(F\otimes\Omega^{1/2}_Y)\to\Gamma^\infty(\psi^*F\otimes\Omega^{1/2}_X)$.
The restriction of $\psi^*$ to sections compactly supported in $\psi(X)$ is invertible and
the inverse is the pushforward $\psi_*:\Gamma_0^\infty((\psi^*F)\otimes\Omega^{1/2}_X)\to \Gamma_0^\infty(F\otimes\Omega^{1/2}_Y)$. 
Then for $u\in\DD'(F\otimes \Omega_Y^{1/2})$, we can define a distribution $\psi^*u\in\DD'(\psi^*F\otimes\Omega_X^{1/2})$
by
\begin{equation}
	(\psi^*u)(f) = u (\psi_* f), \qquad f\in\Gamma_0^\infty((\psi^*F)^*\otimes\Omega^{1/2}_X)
\end{equation}
noting that $(\psi^*F)^*=\psi^* F^*$.
\begin{lemma}
	Defining $\hat{\psi}:T^*X\to T^*Y$ by
	$\hat{\psi}(x,\xi)=(\psi(x),(D\psi|_x^*)^{-1}\xi)$, one has
	\begin{equation}
		\WF_\pol(\psi^* u) = \hat{\psi}^*\WF_\pol(u):= 
		\{(x,\xi;w)\in \pi_X^*\psi^*F: (\hat{\psi}(x,\xi);w)\in\WF_\pol(u)\},
	\end{equation}
	where $\pi_X:T^*X\to X$ is the bundle projection. 
\end{lemma}
\begin{proof} 
	For $x\in X$, choose a local framing of $F\otimes \Omega_Y^{1/2}$ over a chart neighbourhood $(Y_\kappa,\kappa)$ of $\psi(x)$ contained in $\psi(X)$. Pulling back via $\psi$, one obtains a local framing of $\psi^*F\otimes \Omega_X^{1/2}$ over a chart neighbourhood $(\psi^{-1}(Y_\kappa),\kappa\circ\psi)$ of $x$. For $u\in\DD'(F\otimes \Omega_Y^{1/2})$, the polarisation set $\WF_\pol(\psi^*u)$ is determined by considering an intersection over those $A\in\Psi^0_\cl(\psi^*F\otimes\Omega_X^{1/2},\Omega_X^{1/2})$ for which $A\psi^*u$ is smooth, and one may restrict to those $A$ mapping $\Gamma_0^\infty((\psi^*F\otimes\Omega_Y^{1/2})|_{\psi^{-1}(Y_\kappa)})$ to itself. 
	But each such $A$ has a corresponding operator $\psi_*A\psi^*\in  \Psi^0_\cl( F\otimes\Omega_Y^{1/2},\Omega_Y^{1/2})$ that makes $u$ smooth, and whose matrix principal symbol
	is related to that of $A$ by
	\begin{equation}
		\sigma(\psi_*A\psi^*)\circ \hat{\psi} = \sigma(A) 
	\end{equation}
	because the chart representatives $(\psi_*A\psi^*)_{\kappa}$ and $A_{\kappa\circ\psi}$ are equal
	and where the framings of $F\otimes\Omega_Y^{1/2}$ and $\psi^*F\otimes\Omega_X^{1/2}$ mentioned above are used to construct the matrices. It follows that $w\in \WF_\pol(\psi^*u)|_{x,\xi}$ if and only if
	$w\in \WF_\pol(u)|_{\hat{\psi}(x,\xi)}$.	 
\end{proof}

\paragraph{Tensor pullbacks}
Let $T^{r,s}_\CC Y$ be the complexified bundle of rank $\binom{r}{s}$ tensors over $Y$. Then the tangent mapping $D\psi$ determines a map
\begin{equation}
	J^{r,s}=(((D\psi)^*)^{-1})^{\otimes r} \otimes (D\psi)^{\otimes s}:
	\Gamma_0^\infty((T_\CC^{r,s}X)^*)\to \Gamma_0^\infty((T_\CC^{r,s}Y)^*)
\end{equation}
and then permits the construction of the tensor pullback 
$\psi_T^* :\DD'(T^{r,s}_\CC Y\otimes\Omega^{1/2}_Y)\to
\DD'(T^{r,s}_\CC X\otimes\Omega^{1/2}_X)$ by
\begin{equation}
	(\psi^*_T u)(f)= u(\rho_Y^{1/2}J^{r,s}\rho_X^{-1/2}f),
\end{equation}
where $\rho_Y$ is any nonvanishing density on $Y$ and $\rho_X=\psi^*\rho_Y$. One has 
\begin{equation}
	\psi^*_T u = A^{r,s}\psi^* u,
\end{equation}
where $A^{r,s}\in\Psi^0_\cl (\psi^* T^{r,s}_\CC Y\otimes\Omega^{1/2}_X,
T^{r,s}_\CC X\otimes\Omega^{1/2}_X)$ has symbol
\begin{equation}
	a^{r,s}(x,\xi)=((D\psi|_x)^{-1})^{\otimes r} \otimes (D\psi|_x^*)^{\otimes s}.
\end{equation}
Then
\begin{align}
	\WF_\pol(\psi^*_T u) &= \WF_\pol(A^{r,s}\psi^*u)=a^{r,s}\hat{\psi}^*\WF_\pol(u) \nonumber\\ &= 
	\{(x,D\psi|_x^*\zeta;a^{r,s}(x,D\psi|_x^*\zeta)w)\in \pi_X^* T^{r,s}_\CC X: (\psi(x),\zeta;w)\in\WF_\pol(u)\},
\end{align}
which is to say that the wavefront set transforms tensorially under the tensor pullback.
This result extends in an obvious way to bitensorial distributions.

\subsection{Propagation of polarisation for systems of real principal type}\label{sec:prop_pol}

Dencker's main achievement in~\cite{Dencker:1982} was to establish a theorem on the propagation
of polarisations that refines the propagation of singularities theorem due to Duistermaat and H\"ormander~\cite{DuiHoer_FIOii:1972}.

To start, recall that the integral curves of the Hamiltonian vector field $X_q$ lying in the 
zero set $q^{-1}(0)$ of
$q\in C^\infty(T^*X;\RR)$
are called \emph{bicharacteristic strips} for $q$.
Also recall that the \emph{radial vector field} on $T^*X$ is given by $\xi^\mu \partial_{\xi_\mu}$, and a vector field $X$ on $T^*X$ is \emph{nonradial} at points where $X\notin\RR\xi^\mu \partial_{\xi_\mu}$. Adapting Definition 3.1 of~\cite{Dencker:1982} to the bundle setting, a system $P\in \Psi_\cl^m(B)$ with principal symbol $p(x,\xi)$ has \emph{real principal type} at $(y,\eta)\in\dot{T}^*X$ if there is a symbol $\tilde{p}\in S^{m'}(\pi^*\Hom(B))$ and a scalar symbol $q\in S^{m+m'}(T^*X)$ so that $(\tilde{p}p)(x,\xi)=q(x,\xi) \id_{B_x}$ in a neighbourhood of $(y,\eta)$ on which $q$ is real-valued, and 
so that either $q(y,\eta)\neq 0$ or $X_q$ is nonradial at $(y,\eta)$ (the same is true in a neighbourhood of $(y,\eta)$ in either case).

In this situation, let $s_a$ be a local bundle frame for $B$ near $y$. Then~\cite{Dencker:1982} there is a partial connection $D_P$ on $\pi^*B$ along the bicharacteristic strips of $q$ near $(y,\eta)$ defined by 
\begin{equation}
	(D_P w)^a= X_q w^a + \frac{1}{2}(\{\tilde{p},p\}w)^a + \ii \tilde{p}^a_{\phantom{a}b} (p^{\textnormal{sub}})^b_{\phantom{b}c} w^c, 
\end{equation} 
where $w = w^a \pi^* s_a\in\Gamma^\infty(\pi^*B)$, $p^{\textnormal{sub}}$ is the subprincipal symbol of $P$, and $\{\tilde{p},p\}$ is a Poisson bracket. If $c$ is such a bicharacteristic strip, then any line bundle $\RR w$ spanned by a smooth solution $w$ to $D_Pw=0$ along $c$ is called a \emph{Hamilton orbit} of $P$. The Hamilton orbits of $P$ are independent of the choices made in constructing $D_P$ (see Section~4 of~\cite{Dencker:1982}).
We can now state Dencker's result (Theorem~4.2 in~\cite{Dencker:1982}) along with its main consequence.
\begin{thm}\label{thm:prop_pol}
	Let $P\in\Psi_\cl^m(B\otimes\Omega^{1/2})$ and $u\in\DD'(B\otimes\Omega^{1/2})$. Suppose that $P$ has real principal type at $(y,\eta)\in \Char(P)\setminus \WF(Pu)\subset \dot{T}^*X$. Over a neighbourhood of $(y,\eta)$ in $\Char(P)$, $\WF_\pol(u)$ is a union of Hamilton orbits of $P$. Consequently, if $P$ has real principal type along a bicharacteristic strip $c$ for $P$ that does not meet $\WF(Pu)$, then any Hamilton orbit above $c$ is either contained in or disjoint from $\WF_\pol(u)$, and $c$ is either contained in or disjoint from $\WF(u)$.
\end{thm} 
\begin{proof}
	Only the consequence requires justification as the rest is already covered in Theorem~4.2 in~\cite{Dencker:1982}. For $a<b$, let $c:[a,b]\to \dot{T}^*M$ be a bicharacteristic strip not meeting $\WF(Pu)$.   
	Suppose that a given Hamilton orbit $\lambda\mapsto (c(\lambda),w(\lambda))$ above $c$ is not disjoint from $\WF_\pol(u)$, so there exists some $\lambda_0\in [a,b]$ with
	$(c(\lambda_0),w(\lambda_0))\in\WF_\pol(u)$. Then the set of $\lambda$ in $[a,b]$ for which
	$(c(\lambda),w(\lambda))\in\WF_\pol(u)$ is nonempty (as it contains $\lambda_0$), closed (because $\WF_\pol(u)$ is), and
	relatively open (by the first part of Theorem~\ref{thm:prop_pol}); hence it is equal to $[a,b]$
	and the whole Hamilton orbit is contained in $\WF_\pol(u)$. 
	If $c$ is not disjoint from $\WF(u)$ then there is a Hamilton orbit through a point of $\WF_\pol(u)\setminus 0$ above $c$ which must be completely contained in $\WF_\pol(u)$; thus $c\subset\WF(u)$.
\end{proof} 

The last part of Theorem~\ref{thm:prop_pol} is the propagation of singularities theorem for systems of real principal type, generalising the scalar result Theorem~6.1.1 in~\cite{DuiHoer_FIOii:1972} (see also Theorem~2.13 of~\cite{IslamStrohmaier:2020} for a result on the propagation of Sobolev regularity in the bundle context).

The description of Hamilton orbits simplifies when $p=q\id_B$, because the choice $\tilde{p}(x,\xi)=\id_B$ gives
\begin{equation}\label{eq:simple_Dencker_conn}
	(D_P w)^a= X_q w^a  + \ii (p^{\textnormal{sub}})^a_{\phantom{a}b} w^b .
\end{equation}
Then the equation $D_Pw=0$ along a bicharacteristic strip $c$ for $q$ becomes  
a system of ordinary differential equations for the frame components of $W=w\circ c$,
\begin{equation}\label{eq:DPWeq0}
	\dot{W}^a + \ii (p^{\textnormal{sub}}\circ c)^a_{\phantom{a}b} W^b =0.
\end{equation}  
In the next section, these results will be applied to operators $P$, $P\otimes 1$, $1\otimes\sadj{P}$ and $P\otimes 1-1\otimes\sadj{P}$, where $P$ is normally hyperbolic.

\section{Polarisation sets of Green operators}\label{sec:polGreen}

\subsection{Normally hyperbolic operators} 
\label{sec:nhypexamples} 

Let $(M,g)$ be a globally hyperbolic spacetime with Levi-Civita connection $\nabla$.
Consider first the operator $\Box=g^{\mu\nu}\nabla_\mu \nabla_\nu$ on $\Gamma^\infty(\Omega^{1/2})$.
In local coordinates, and writing $\partial_\mu=\partial/\partial x^\mu$,
\begin{align}\label{eq:densitybox}
	\Box &= (-g)^{-1/4}\partial_\mu g^{\mu\nu} (-g)^{1/2}\partial_\nu (-g)^{-1/4}  \\
	&= g^{\mu\nu}\partial_\mu\partial_\nu + (\partial_\mu g^{\mu\nu})\partial_\nu + \textnormal{zero order},
\end{align} 
so the full symbol is $-g^{\mu\nu} k_\mu k_\nu+\ii  (\partial_\mu g^{\mu\nu})k_\nu$ up to zeroth order terms. The principal symbol is $q(x,k)=-  g^{\mu\nu} k_\mu k_\nu=-g^{-1}(k,k)$, so $\Char(P)=\Nc$,
while the refined principal symbol is
\begin{equation}
	q^r(x,k) = -g^{\mu\nu} k_\mu k_\nu+\ii  (\partial_\mu g^{\mu\nu})k_\nu 
	+\frac{\ii}{2}\frac{\partial^2q}{\partial x^\mu \partial k_\mu} = q(x,k),
\end{equation}
and the subprincipal symbol is thus $q^{\textnormal{sub}}(x,k)=0$. The Hamiltonian vector field determined by $q$ is 
\begin{equation}
	X_q = -2g^{\mu\nu}k_\nu \frac{\partial}{\partial x^\mu} + k_\alpha k_\beta \frac{\partial g^{\alpha\beta}}{\partial x^\mu}\frac{\partial}{\partial k_\mu},
\end{equation}
and its integral curves take the form $c(\lambda)=(x(\lambda),k(\lambda))$, where $x(\lambda)$ is an affinely parametrised geodesic and $k=-\tfrac{1}{2}\dot{x}^\flat$ is invariant under parallel transport along $x$, and cotangent to $x$; $c$ is a bicharacteristic strip precisely when $x(\lambda)$ is a null geodesic.

Now consider a normally hyperbolic operator $P$ on $\Gamma^\infty(B\otimes\Omega^{1/2})$, where
$B$ is a finite-rank complex vector bundle over $M$. Then $(-g)^{-1/4}P(-g)^{1/4}$ 
is a normally hyperbolic operator on $\Gamma^\infty(B)$ and may be written in Weitzenb\"ock form
\begin{equation}\label{eq:Weitzenboeck}
	(-g)^{-1/4}P(-g)^{1/4} = \Box^B + V, \qquad \Box^B:=g^{\mu\nu} \nabla_\mu^{T^*M\otimes B}\nabla_\nu^B  
\end{equation}
where $V\in \Gamma^\infty(\End(B))$ and $\nabla^B$ is the Weitzenb\"ock connection on $B$. Introduce a local frame $s_a$ for $B$ and define connection $1$-forms obeying
\begin{equation}
	s_a\Gamma^B_\mu{}^{a}_{\phantom{a}b} = \nabla_\mu^B s_b.
\end{equation}
A short calculation shows that $(Pf)^a=P^a_{\phantom{a}b}f^b$, where
the operators $P^a_{\phantom{a}b}$ on $\Gamma^\infty(\Omega^{1/2})$ 
are 
\begin{equation}
	P^{a}_{\phantom{a}b} = \delta^{a}_{\phantom{a}b} \Box + 2 g^{\mu\nu}\Gamma^B_\mu{}^{a}_{\phantom{a}b}\partial_\nu + \hat{V}^{a}_{\phantom{a}b}
\end{equation}
for some possibly-modified zero-order part $\hat{V}^{a}_{\phantom{a}b}$, and where $\Box$ acts on $\Gamma^\infty(\Omega^{1/2})$ as in~\eqref{eq:densitybox}. The principal symbol is $p(x,k)=q(x,k)\id_{B_x}$, so the subprincipal symbol can be understood as a partial connection as described earlier. 

Because $q^{\textnormal{sub}} = 0$, and the additional contributions to $P^{a}_{\phantom{a}b}$ are of first order, it follows that
\begin{equation}
	(p^{\textnormal{sub}})^{a}_{\phantom{a}b}(x,k) = 
	2\ii g^{\mu\nu}k_\nu\Gamma^B_\mu{}^{a}_{\phantom{a}b} = - \ii (\Gamma_{X_q}^{\pi^*B})
	{}^{a}_{\phantom{a}b},
\end{equation}
where $X_q$ is the Hamiltonian vector field obtained above.
In particular the invariantly defined expression $p^{\textnormal{sub}}+\ii\Gamma_{X_q}^{\pi^*B}$ vanishes. Our result here is closely related to the discussion of $P$-compatible connections in Section~2.2, and particularly Theorem 2.8, of~\cite{IslamStrohmaier:2020} and we have benefited from their presentation (of which ours may be considered a reworking). Similar results have been obtained in more specific contexts, including the Laplace--Beltrami operator on arbitrary rank covectors~\cite{Hintz:2017}, and the $1$-form Maxwell equation and Dirac equations~\cite{Kratzert:2000,Hollands:2001}. For future reference, we note the formula
\begin{equation}\label{eq:psub_simple}
	\ii p^{\textnormal{sub}} (x,k) = \Gamma^{\pi^*B}_{X_q}(x,k) = \Gamma^B_{\pi_*X_q}(x) = \Gamma^{B}_{-2k^\sharp}(x) ,
\end{equation}
in which indices have been suppressed. 

Next, consider the operators $P\otimes 1$ and $1\otimes\sadj{P}$ acting on 
$\Gamma^\infty((B\boxtimes B^*)\otimes\Omega^{1/2}_{M\times M})$,  where $\sadj{P}$ is the formal dual to $P$ and has Weitzenb\"ock form
\begin{equation}
	\sadj{P} = \Box^{B^*} + \sadj{V},
\end{equation}
as is seen by a short calculation.
The principal symbols are 
\begin{equation}
	\sigma(P\otimes 1)(x,k,y,l)= q(x,k)\id_{B_x\otimes B^*_y}, \qquad
	\sigma(1\otimes \sadj{P})(x,k,y,l)= q(y,l)\id_{B_x\otimes B^*_y},
\end{equation}
so 
\begin{equation}
	\Char P\otimes 1 = (\Nc_0\times T^*M)\setminus 0, \qquad \Char 1\otimes\sadj{P}= (T^*M\times \Nc_0)\setminus 0,
\end{equation}
and the corresponding scalar symbols are $q\otimes 1$ and $1\otimes q$.  
As a derivation on $C^\infty(T^*(M\times M))$, regarded as the completed tensor product
$C^\infty(T^*M) \mathbin{\widehat{\otimes}} C^\infty(T^*M)$, the symbol $q\otimes 1$ has Hamiltonian vector field $X_q\otimes 1$, with integral 
curves in $\dot{T}^*M\times T^*M$ taking the form $\lambda\mapsto (c(\lambda);y,l)$, where $c$ is an integral curve of $X_q$. The bicharacteristic strips for $q\otimes 1$ are the integral curves of this type for 
which $c$ is a bicharacteristic strip for $q$, together with the somewhat trivial bicharacteristic strips $\lambda\mapsto (x,0;y,l)$ in $0\times\dot{T}^*M$. Similarly, the Hamiltonian vector field $1\otimes X_q$ of $1\otimes q$ generates integral 
curves (resp., bicharacteristic strips) in $T^*M\times \dot{T}^*M$ of the form $\lambda\mapsto (x,k;c(\lambda))$, where $c$ is an integral curve of $X_q$ (resp., bicharacteristic strip for $q$) 
and trivial bicharacteristic strips $\lambda\mapsto (x,k;y,0)$ in $\dot{T}^*M\times 0$.
Choosing a frame for $B$ and a dual frame for $B^*$, the matrices of
subprincipal symbols at $(x,k;y,l)\in T^*(M\times M)$ are
$p^\textnormal{sub}(x,k)\otimes \id$ and $\id \otimes \sadj{p}^\textnormal{sub}(y,l)$.  

Finally, we also consider the difference $P\otimes 1-1\otimes\sadj{P}$, with scalar principal symbol $q\otimes 1- 1\otimes q$,  and characteristic set 
\begin{equation}
	\Char(P\otimes 1-1\otimes\sadj{P}) = \{(x,k;x',k')\in\dot{T}^*(M\times M): g^{-1}_x(k,k)=g^{-1}_{x'}(k',k')\},
\end{equation}
which includes the set $(0\times\Nc)\cup(\Nc\times 0)$. The Hamiltonian vector field is  $X_q\otimes 1-1\otimes X_q$. Because $q$ is quadratic in momenta, the integral curves in $\dot{T}^*M\times\dot{T}^*M$ may be expressed in the form $\lambda\mapsto(c(\lambda), -c'(\lambda))$ where $c$ and $c'$ are integral curves of $X_q$ and
the minus sign is understood as usual in $T^*M$; among these, the integral curves in which $c$ and $c'$ are over geodesics of the same causal type and with the same `energy' are bicharacteristic strips.
There are also degenerate bicharacteristic strips of the form $c\times (x',0)$ or $(x,0)\times (-c')$, where
$c$ and $c'$ are bicharacteristic strips of $q$. 
The subprincipal symbol at $(x,k;y,l)\in T^*(M\times M)$  is $p^\textnormal{sub}(x,k)\otimes \id - \id\otimes \sadj{p}^\textnormal{sub}(y,l)$.

Finally, if $P$ is normally hyperbolic then it has real principal type everywhere in $\dot{T}^*M$, while $P\otimes 1$ (resp., $1\otimes\sadj{P}$) has real principal type on $\dot{T}^*M\times T^*M$ (resp., $T^*M\times\dot{T}^*M$) and $P\otimes 1-1\otimes \sadj{P}$ has real principal type
everywhere in $\dot{T}^*(M\times M)$. In each case this is seen because the relevant Hamiltonian vector field is nonradial on the given sets.

\subsection{Hamilton orbits for normally hyperbolic and related operators}

Continuing with the above notation, let us now describe the Hamilton orbits of the operators $P$, $P\otimes 1$, $1\otimes \sadj{P}$ and $P\otimes 1-1\otimes\sadj{P}$ that will be needed in the proof of
Theorem~\ref{thm:WFpolEPpmandEP}.

The subprincipal symbol of $P$ satisfies $\ii p^{\textnormal{sub}}(c(\lambda))= \Gamma^{\pi^*B}_{\dot{c}(\lambda)}(c(\lambda))=\Gamma^B_{\dot{x}(\lambda)}(x(\lambda))$ (in any frame)
along any bicharacteristic strip $c$, where
$x(\lambda)=\pi(c(\lambda))$ and $\dot{x}=\pi_*\dot{c}$. Consequently,~\eqref{eq:DPWeq0} is equivalent to the equation of parallel transport
\begin{equation}
	\nabla^{\pi^*B}_{\dot{c}} w = 0,
\end{equation}
along $c$. However,~\eqref{eq:DPWeq0} is also equivalent to the equation of $\nabla^B$-parallel transport along the null geodesic $x=\pi\circ c$. As $x$ never revisits any point of $M$, it follows that the Hamilton orbits of $P$ are spanned by curves
$\lambda\mapsto (x(\lambda),k(\lambda);w(\lambda))$ in $\pi^*B$, where
$x$ is an affine null geodesic, $k(\lambda)=-\tfrac{1}{2}\dot{x}(\lambda)^\flat$, 
and $w$ is $\nabla^B$-parallel. Thus, 
\begin{equation}
	w(\lambda) = \Pi^{x(\lambda_0),k(\lambda_0)}_{x(\lambda),k(\lambda)} w(\lambda_0),
\end{equation}
where, as in the introduction, $\Pi^{x',k'}_{x,k}$ denotes the operator of $\nabla^B$-parallel transport
from $x'$ to $x$ along the unique witnessing null geodesic if $(x,k)\sim (x',k')$, and
$\Pi_{x,k}^{x',k'}=0$ otherwise.

Turning to the operator $P\otimes 1$, the Hamilton orbits over $\dot{T}^*M\times T^*M$ are spanned by curves $\lambda\mapsto (c(\lambda);y,l;z(\lambda))$ in $(\pi\times\pi)^*(B\boxtimes B^*)$, where $c(\lambda)=(x(\lambda),k(\lambda))$ is as before and the partial section $z$ obeys $D_{P\otimes 1}z=0$
along $\lambda\to C(\lambda)= (c(\lambda);y,l)$. Writing 
$Z=z\circ C$, the 
analogue of~\eqref{eq:DPWeq0} is
\begin{equation} 
	\dot{Z}(\lambda)  + \ii (p^{\textnormal{sub}}(c(\lambda))\otimes \id)  Z(\lambda) =0,
\end{equation}
in a frame for $B\boxtimes B^*$, but with indices suppressed, i.e., with $x=\pi\circ c$,
\begin{equation}
	\dot{Z}(\lambda) + (\Gamma_{\dot{x}}^B(x(\lambda)) \otimes \id)Z(\lambda) =0,
\end{equation}
which is solved by 
\begin{equation}
	z(\lambda)=(\Pi^{c(\lambda_0)}_{c(\lambda)}\otimes \id)z(\lambda_0).
\end{equation}
Similarly, the Hamilton orbits of $1\otimes\sadj{P}$ in $T^*M\times \dot{T}^*M$ are spanned by
curves $\lambda\mapsto (y,l;c(\lambda);z(\lambda))$ with
\begin{equation}
	z(\lambda)=(\id\otimes\sadj{\Pi}{}^{c(\lambda_0)}_{c(\lambda)})z(\lambda_0),
\end{equation}
where $\sadj{\Pi}{}^{x',k'}_{x,k}$ is defined analogously to $\Pi^{x',k'}_{x,k}$ but for $\nabla^{B^*}$-parallel transport. 

Finally, consider the Hamilton orbits of $P\otimes 1-1\otimes\sadj{P}$ 
in $\dot{T}^*{M}\times \dot{T}^*M$ (we will not need the orbits
in $(0\times \dot{T}^*M)\cup (\dot{T}^*M\times 0)$). These are
spanned by curves $\lambda\mapsto (c(\lambda);-c'(\lambda);z(\lambda))$, where $c$ and $c'$ are integral curves of $X_q$ in $\dot{T}^*M$ with equal energy (see Section~\ref{sec:nhypexamples}), and $z$ solves $D_{P\otimes 1-1\otimes\sadj{P}}z=0$ along
$c\times -c'$. Then $Z(\lambda)=z(c(\lambda),-c'(\lambda))$ satisfies
\begin{equation}
	\dot{Z}(\lambda) +\ii (p^{\textnormal{sub}}(c(\lambda))\otimes \id - \id\otimes \sadj{p}^\textnormal{sub}(-c'(\lambda)))  Z(\lambda) =0,  
\end{equation}
in a frame for $B\boxtimes B^*$, with indices suppressed as before. 
Using~\eqref{eq:psub_simple}, 
\begin{equation}
	\ii\, \sadj{p}^\textnormal{sub}(-c'(\lambda))=\Gamma^{B^*}_{2k'(\lambda)^\sharp}(x'(\lambda))
	=-\Gamma^{B^*}_{\dot{x}'(\lambda)}(x'(\lambda))
\end{equation} 
because $\dot{x} = -2k^\sharp$ along integral curves of $X_q$.
Thus $Z$ solves
\begin{equation}
	\dot{Z}(\lambda) + (\Gamma^B_{\dot{x}(\lambda)}(x(\lambda))\otimes \id + \id\otimes \Gamma^{B^*}_{\dot{x}'(\lambda)}(x'(\lambda)))  Z(\lambda) =0,  
\end{equation}
and the solutions obey
\begin{equation}\label{eq:Zprop}
	Z(\lambda) = 
	(\Pi{}^{c(\lambda_0)}_{c(\lambda)}\otimes\sadj{\Pi}{}^{-c'(\lambda_0)}_{-c'(\lambda)})Z(\lambda_0),
\end{equation} 
where we have temporarily abused the notation for the parallel propagator to cover timelike and spacelike geodesics as well. However, we will only use~\eqref{eq:Zprop} when $x$ and $x'$ are null geodesics. 
Note that the minus sign in~\eqref{eq:Zprop} is a matter of convention, because $\sadj{\Pi}{}^{-c'(\lambda_0)}_{-c'(\lambda)}=\sadj{\Pi}{}^{c'(\lambda_0)}_{c'(\lambda)}$, but this seems a natural choice because~\eqref{eq:Zprop} can also be interpreted as parallel transport 
of $(\pi\times\pi)^*z$ along
$c\times -c'$ with respect to $\nabla^{(\pi\times\pi)^* B\boxtimes B^*}$.

In short, we have seen how the propagation of polarisation for the operators $P$, $P\otimes 1$, $1\otimes \sadj{P}$ and $P\otimes 1-1\otimes \sadj{P}$ can be understood in terms of parallel transport along suitable curves
under the Weitzenb\"ock connection for the operator $(-g)^{-1/4}P(-g)^{1/4}$ and connections constructed from it. We now use these results to determine the polarisation sets for $E_P^\pm$ and $E_P$.

\subsection{Proof of Theorem~\ref{thm:WFpolEPpmandEP} and Corollary~\ref{cor:WFpol2}}

We recall the main assumptions of Theorem~\ref{thm:WFpolEPpmandEP}.
Let $P$ be a normally hyperbolic operator on $\Gamma^\infty(B\otimes\Omega^{1/2})$, where $B$ is a finite-rank complex vector bundle over globally hyperbolic spacetime $(M,g)$. Let
$\nabla^B$ be the Weitzenb\"ock connection on $B$ determined by $(-g)^{-1/4}P(-g)^{1/4}$ and define
$\Rc_\pol^\pm$, $\Rc_\pol$ as in~\eqref{eq:Rcpolsets}. Our task is to compute $\WF_\pol(E_P^\pm)$ and
$\WF_\pol(E_P)$. There is a special case in which $\WF_\pol(E_P^\pm)$ are easily computed. 
\begin{lemma}\label{lem:diagonalcase}
	Under the hypotheses of Theorem~\ref{thm:WFpolEPpmandEP}, 
	suppose that $B$ admits a global frame of smooth sections $s_a$ and a global dual frame $t^a$ for $B^*$, with respect to which $P$ takes the diagonal form
	\begin{equation}
		P^a_{\phantom{a}b}=\delta^a_{\phantom{a}b} Q
	\end{equation}
	where $Q$ is a normally hyperbolic operator on $\Gamma^\infty(\Omega^{1/2})$. Then $\WF_\pol(E_P^\pm)=\Rc_\pol^\pm\cup \WF_\pol(\id^\knl)$.
\end{lemma}
\begin{proof}
	The hypothesis implies that
	$P (f^a s_a) = (Q f^a)s_a$ for coefficient half-densities $f^a\in \Gamma^\infty(\Omega^{1/2})$; consequently, the Green operators of $P$ are
	\begin{equation}
		E_P^\pm = s_a E_Q^\pm t^a
	\end{equation}
	with distributional kernel
	\begin{equation}
		E_P^\pm(x,y) = E_Q^\pm(x,y) s_a(x)\otimes t^a(y),
	\end{equation}
	recalling that $E_Q^\pm\in\DD'(\Omega^{1/2}_{M\times M})$.
	As $\Pi(x,y)=s_a(x)\otimes t^a(y)$ defines a nonvanishing section $\Pi$ of $B\boxtimes B^*$, we have by Lemma~\ref{lem:WFpol1} that
	\begin{equation}\label{eq:WFpolEPpm_diagonal}
		\WF_\pol(E^\pm_P)  = \{(x,k;x',-k';w)\in (\pi\times\pi)^*(B\boxtimes B^*) : (x,k;x',-k')\in\WF(E^\pm_P),~
		w\in\CC\Pi(x,x')\}\cup 0.
	\end{equation}
	In particular, note that $\Pi(x,x)=\delta_x$, the image of $\id_{B_x}$ under the identification of $\Lin(B_x)$ with $B_x\otimes B_x^*$.
	
	Now define $\tilde{P}=(-g)^{-1/4}P(-g)^{1/4}$, $\tilde{Q}=(-g)^{-1/4}Q(-g)^{1/4}$. 
	Writing $\tilde{P}$ in Weitzenb\"ock form, and using the normal hyperbolicity of $Q$ and hence $\tilde{Q}$, one has
	\begin{equation}
		(\Box f)s_a + 2\nabla^B_{(\nabla f)^\sharp}s_a + f\tilde{P}s_a = 
		(\Box f + 2\nabla_Z f+ Wf)s_a, \qquad f\in C^\infty(M)
	\end{equation}
	for some smooth vector field $Z$ and scalar field $W$. Considering $f\equiv 1$, 
	one sees that $\tilde{P}s_a = Ws_a$, and we may thus simplify to obtain
	\begin{equation}
		\nabla^B s_a = Z^\flat \otimes s_a.
	\end{equation}
	It follows that $s_a$ is $\nabla^B$-parallel along any curve, up to a nonzero scalar integrating factor (which can depend on the curve). Applied to the witnessing geodesic
	for the relation $(x,k)\sim (x',k')$, this yields $\Pi(x,x')\in \CC\Pi^{x',k'}_{x,k}$, which concludes the proof. 
\end{proof}

The proof of Theorem~\ref{thm:WFpolEPpmandEP} employs the following strategy. Using propagation of polarisation for $P\otimes 1$ and $1\otimes\sadj{P}$, we show how $\WF_\pol(E_P^\pm)$ is determined by its fibres over the twisted diagonal $\{(x,k;x,-k)\in \dot{T}^*(M\times M) :(x,k)\in\dot{T}^*M\}$ of $\dot{T}^*(M\times M)$. 
We show that such fibres either belong to $\WF_\pol(\id)$ or lie over points
$(x,k;x,-k)$ with $(x,k)\in\Nc$. In the latter case, we construct a normally hyperbolic operator $\tilde{P}$ that agrees with $P$ near $x$, but is diagonal near some $y$,
where $(y,l)$ lies on the bicharacteristic strip through $(x,k)$. Then the polarisation sets of $E_P^\pm$ and $E_{\tilde{P}}^\pm$ agree over $(x,k;x,-k)$, and this fibre 
can be connected to the fibre over $(y,l;y,-l)$ (given by Lemma~\ref{lem:diagonalcase})
using the propagation of polarisation
for $\tilde{P}\otimes 1-1\otimes\sadj{\tilde{P}}$. This results in an upper bound $\WF_\pol(E_P^\pm)\subset \Rc_\pol^\pm\cup \WF_\pol(\id)$. The reverse inclusion is shown by another
propagation of singularities argument.
\begin{proof}[Proof of Theorem~\ref{thm:WFpolEPpmandEP}] 
	(a) We compute $\WF_\pol(E_P^+)$ in four steps (the computation of $\WF_\pol(E_P^-)$ is analogous). First, the identities
	\begin{equation}\label{eq:kernel_identities}
		(P\otimes 1)E_P^+ = \id = (1\otimes \sadj{P}) E_P^+, 
	\end{equation} 
	together with the fact that $P$ and $\sadj{P}$ have scalar principal symbols,  
	give
	\begin{equation}\label{eq:step1a}
		\WF_\pol(E_P^+)\subset
		\WF_\pol(\id)\cup (\pi\times\pi)^*(B\boxtimes B^*)|_{(\Nc_0\times\Nc_0)\setminus 0},
	\end{equation}
	where we have used~\eqref{eq:WFpolgenbds}.
	Moreover, the support bound $\supp E^+_P f\subset J^+(\supp f)$
	implies that 
	\begin{equation}\label{eq:step1b}
		\WF(E_P^+)\subset \{(x,k;x',-k')\in \dot{T}^*(M\times M): x\in J^+(x')\}.
	\end{equation}	 
	
	Second, we use propagation of polarisation to obtain a global upper bound for
	$\WF_\pol(E_P^+)$ in terms of its fibres over the twisted diagonal. 
	Suppose that $(x,k;x',-k';w)\in \WF_\pol(E_P^+)\setminus \WF_\pol(\id)$ with $k\neq 0$.
	Then by~\eqref{eq:step1a} we have $(x,k)\in\Nc$, $(x',k')\in\Nc_0:=\Nc\cup 0$.
	It follows that if $c$ is the unique complete bicharacteristic strip for $P$ through $(x,k)$
	then $c\times (x',-k')$ is a bicharacteristic strip for 
	$P\otimes 1$, along which it has real principal type (see the end of Section~\ref{sec:nhypexamples}). If $c\times(x',-k')$ does not meet $\WF((P\otimes 1)E_P^+)=\WF(\id)$, i.e., the twisted diagonal,
	then $c\times(x',-k')$ is wholly contained in $\WF(E_P^+)$
	by the propagation of singularities part of
	Theorem~\ref{thm:prop_pol} applied to $P\otimes 1$. However $\pi\circ c$ includes points outside $J^+(x')$, contradicting~\eqref{eq:step1b}. Accordingly, $c$ must meet $(x',k')$, thus giving $(x,k)\sim(x',k')$ and in particular $k'\neq 0$,	
	so $(x,k;x',-k')\in\Rc^+$. Furthermore, Theorem~\ref{thm:prop_pol} also implies that
	$(\Pi^{x,k}_{x',k'}\otimes 1)w \in \WF_\pol(E_P^+)|_{(x',k',x',-k')}$ using propagation of polarisation and the closure of the polarisation set. Inverting the parallel transport operator, we see that $w\in (\Pi^{x',k'}_{x,k}\otimes 1)\WF_\pol(E_P^+)|_{(x',k';x',-k')}$. 
	
	A similar argument for the operator $1\otimes\sadj{P}$ shows that $\WF_\pol(E_P^+)\setminus \WF_\pol(\id)$ contains no
	points of the form $(x,0;x',-k';w)$ with $k'\neq 0$. Thus, we have shown that  
	\begin{equation}\label{eq:step2}
		\WF_\pol(E_P^+)\subset \WF_\pol(\id)\cup
		\{ (\Pi^{x',k'}_{x,k}\otimes 1)\WF_\pol(E_P^+)|_{(x',k';x',-k')},~
		(x,k;x',-k')\in \Rc^+
		\}.
	\end{equation} 
	
	The third step refines this upper bound using a deformation argument in combination with propagation of polarisation to show that $\WF_\pol(E_P^+)|_{(x,k;x,-k)}\subset \CC\delta_x$ for any $(x,k)\in\Nc$.
	For suppose $w\in \WF_\pol(E_P^+)|_{(x,k;x,-k)}$.
	Choose a region (i.e., an open causally convex subset of $M$) $N$ containing $x$ over which there is a local frame $s_a$ for $B|_N$ and corresponding dual frame $t^a$ for $B^*|_N$. 
	In frame components, $P$ is represented by a matrix
	of operators on $\Gamma^\infty(\Omega^{1/2})$ given by 
	\begin{equation}
		P^a_{\phantom{a}b}=\delta^a_{\phantom{a}b} \Box + 
		A^{\mu a}_{\phantom{\mu a}b}\nabla_\mu +
		V^{a}_{\phantom{a}b},
	\end{equation}
	where $\mu$ is an abstract spacetime index, and $A^{\mu a}_{\phantom{\mu a}b}$ is a matrix of smooth vector fields over $N$, while $V^{a}_{\phantom{a}b}$ is a matrix of smooth functions over $N$.
	
	Now choose Cauchy surfaces $\Sigma^\pm$ for $N$, with $\Sigma^\pm\subset I^\pm(\Sigma^\mp)$ and $x\in I^+(\Sigma^+)$. Choose $\chi\in C^\infty(N)$ with $\chi=1$ on $N^+=I^+(\Sigma^+)$ and $\chi=0$ on $N^-=I^-(\Sigma^-)$, and define a 
	normally hyperbolic $\tilde{P}$ on $N$ with frame components
	\begin{equation}
		\tilde{P}^a_{\phantom{a}b}=\delta^a_{\phantom{a}b} \Box + 
		\chi A^{\mu a}_{\phantom{\mu a}b}\nabla_\mu +
		\chi V^{a}_{\phantom{a}b}.
	\end{equation}
	As $N^\pm$ are regions of $N$, the distributional
	kernels of $E^+_{P}$ and $E^+_{\tilde{P}}$ coincide on $N^+\times N^+$, while the diagonal form of $\tilde{P}$ in $N^-$ fixes $\WF_\pol(E^+_{\tilde{P}})$ over $N^-\times N^-$ by Lemma~\ref{lem:diagonalcase}.
 	
	As $(x,k)\in\Nc$, there exists $(y,l)\sim (x,k)$ with $y\in N^-$
	and witnessing bicharacteristic strip $c$; moreover, we have 
	$(x,k;x,-k;w)\in \WF_\pol(E_{\tilde{P}}^+)\setminus 0$. 
	Using $(\tilde{P}\otimes 1-1\otimes\sadj{\tilde{P}})E_{\tilde{P}}^+=0$, 
	propagation of polarisation along $c\times (-c)$ entails that
	\begin{equation}
		\left(y,l;y,-l; (\Pi^{x,k}_{y,l}\otimes \sadj{\Pi}^{x,-k}_{y,-l})w\right)\in \WF_\pol(E_{\tilde{P}}^+)
	\end{equation}
	and hence, by Lemma~\ref{lem:diagonalcase},
	\begin{equation}
		w \in \CC (\Pi^{y,l}_{x,k}\otimes \sadj{\Pi}^{y,-l}_{x,-k})\delta_y = \CC \delta_x.
	\end{equation}
	The last identity is proved by choosing a frame $s_a$ for $B_y$ and a dual frame $t^a$ for $B_y^*$ and extending by parallel transport along the witnessing geodesic with respect to $\nabla^B$ and $\nabla^{B^*}$ respectively, so that the frames remain dual. Noting that $\delta_z=s_a(z)\otimes t^a(z)$ holds for any $z$ on the witnessing geodesic, and 
	$\Pi^{y,l}_{x,k}s_a(y)=s_a(x)$, $\sadj{\Pi}^{y,-l}_{x,-k}t^a(y)=
	\sadj{\Pi}^{y,l}_{x,k}t^a(y)=t^a(x)$, one has the required identity. 
	
	We have now established that $\WF_\pol(E_P^+)|_{(x,k;x,-k)}\subset \CC\delta_x$ for all $(x,k)\in\Nc$.
	In combination with~\eqref{eq:step2} and the fact that $(\Pi^{x',k'}_{x,k}\otimes 1)\delta_{x'}=\Pi^{x',k'}_{x,k}$,	
	this provides the upper bound 
	\begin{equation}\label{eq:WFpolupperbd}
		\WF_\pol(E_P^+)\subset 	 \WF_\pol(\id)\cup \Rc_\pol^+. 
	\end{equation}
	Now the fibre of the right-hand side over $\WF(\id)$ is at most one-dimensional, but on the other hand, $\WF(\id)=\WF((P\otimes 1)E_P^+)\subset\WF(E_P^+)$ by~\eqref{eq:WFPu}. Therefore we deduce that 
	$\WF_\pol(\id)\subset \WF_\pol(E_P^+)$.
	
	The fourth step, illustrated in Fig.~\ref{fig:fourth_step} establishes that $\Rc_\pol^+\subset \WF_\pol(E_P^+)$, completing the proof. 
	Take any $(x,k;x',-k')\in\Rc^+$ with $x\neq x'$ 
	and let $c$ be the complete bicharacteristic strip
	for $P$ witnessing to $(x,k)\sim (x',k')$. Construct a region $N$ containing $x$ and
	deformed operator $\tilde{P}$ as in the third step, and pick $(x'',k'')$ on $c$ so that
	$x''\in N^+$ and $x\in J^+(x'')$. (It will not matter whether $x''$ lies to the future or past of $x'$, and the possibility $x''=x'$ is not excluded.)  
	 Also choose distinct points $(y,l)$ and $(y',l')$
	on $c$ with $y,y'\in N^-$ and $y\in J^+(y')$. Then $(y,l;y',-l')\in \WF(E_{\tilde{P}}^+)$ and by propagation of singularities for $\tilde{P}\otimes 1$ this also implies $(x,k;y',-l')\in \WF(E_{\tilde{P}}^+)$, noting that $c$ does not meet $(y',l')$ between $(y,l)$ and $(x,k)$ (see Fig.~\ref{fig:fourth_step}(b)). Similarly, propagation of singularities for
	$1\otimes\sadj{\tilde{P}}$ implies that $(x,k;x'',-k'')$ lies in $\WF(E_{\tilde{P}}^+)$ and hence $\WF(E_{P}^+)$, because $-c$ does not meet
	$(x,-k)$ between $(y',-l')$ and $(x'',-k'')$ (Fig.~\ref{fig:fourth_step}(c)). Lastly, propagation of singularities
	for $1\otimes\sadj{P}$ implies that $(x,k;x',-k')\in\WF(E_P^+)$ because
	$-c$ does not meet $(x,-k)$ between $(x'',-k'')$ and $(x',-k')$ (Fig.~\ref{fig:fourth_step}(d)).
	Hence $\WF(E_P^+)$ contains every $(x,k;x',-k')\in\Rc^+$ with
	$x'\neq x$ and hence $\Rc^+\subset\WF(E_P^+)$ by closure of the wavefront set. Thus $\WF_\pol(E_P^+)$ has nontrivial fibre above every point of $\Rc^+$ and since
	the upper bound~\eqref{eq:WFpolupperbd} shows that the fibre is at most one-dimensional, it must equal $\CC\Pi^{x',k'}_{x,k}$. 
	Accordingly, $\Rc_\pol\subset\WF_\pol(E_P^+)$ and the proof of~(a) for $E_P^+$ is complete. The computation of $\WF_\pol(E_P^-)$ is directly analogous. 
	
	 \newcommand{\makeAlph}[1]{%
	 		\ifcase #1?\or a\or b\or c\or d\or e\or f\or g\or h\or
	 		i\or j\or k\or l\or m\or n\or o\or p\or q\or r\or
	 		s\or t\or u\or v\or w\or x\or y\or z\else ?\fi}
	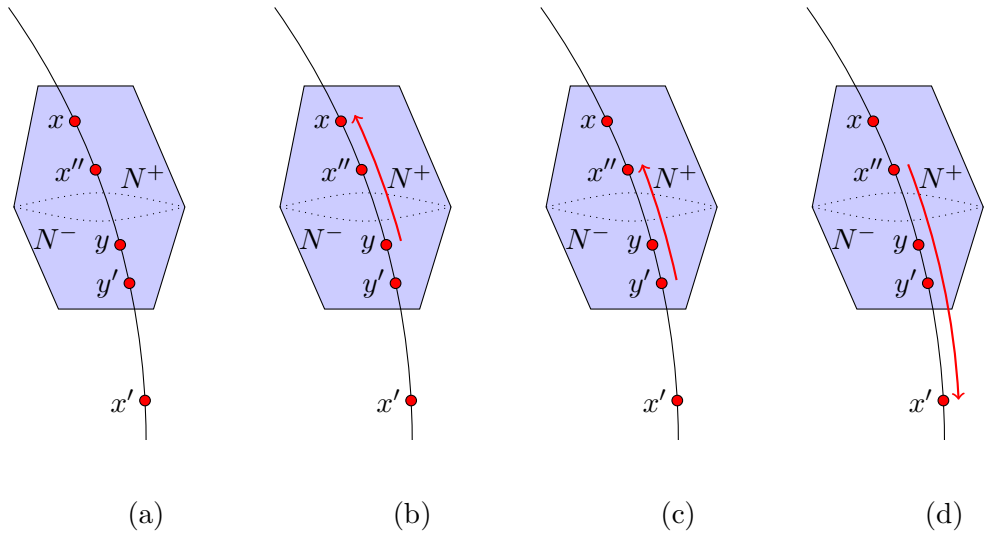
\begin{figure}
		\begin{center}
	\begin{tikzpicture}
	\foreach \x in {1,2,3,4}
	{\begin{scope}[xshift =3.5*\x cm]
	\clip (-2,-1.5) rectangle (2.5,8);
	\path (0,0) arc (0:28:10) coordinate (A);
	\path (0,0) arc (0:10:10) coordinate (B);
	\path (0,0) arc (0:18:10) coordinate (C);
	\draw[fill=blue!20] (A) -- ++(1,0) -- ($ (C)+(1,0) $) -- ($ (B)+(0.25,0) $) -- ++(-1.25,0) -- ($ (C)+(-1.25,0) $) -- ($ (A) +(-0.25,0)$) -- cycle;  
	\draw (C)++(-0.7,-0.4) node {$N^-$};
	\draw (C)++(0.45,0.4) node {$N^+$};
	\draw (0,0) arc (0:35:10); 
	\draw[dotted] ($ (C) +(-1.25,0) $) .. controls ++(1.125,0.25) .. ($ (C) +(1,0) $);
	\draw[dotted] ($ (C) +(-1.25,0) $) .. controls ++(1.125,-0.25) .. ($ (C) +(1,0) $);
	\filldraw[fill=red]  (-10,0)+(25:10cm) circle (2pt) node[left]{$x$};
	\filldraw[fill=red]  (-10,0)+(21:10cm) circle (2pt) node[left]{$x''$};
	\filldraw[fill=red]  (-10,0)+(15:10cm) circle (2pt) node[left]{$y$};
	\filldraw[fill=red]  (-10,0)+(12:10cm) circle (2pt) node[left]{$y'$};
	\filldraw[fill=red]  (-10,0)+(3:10cm) circle (2pt) node[left]{$x'$}; 
	\ifcase \x \relax\or\relax\or 
	\draw[->,red,thick] (-10,0)+(15:10.2cm) arc (15:25:10.2);\or
	\draw[->,red,thick] (-10,0)+(12:10.2cm) arc (12:21:10.2);\else
	\draw[->,red,thick] (-10,0)+(21:10.2cm) arc (21:3:10.2);\fi
	\draw (0,-1) node {(\makeAlph{\x})};
	\end{scope}} 
	\end{tikzpicture}
	\end{center}
	\caption{(a) The points $x$, $x'$, $y$, $y'$ and $x''$ along a null geodesic and region $N$ (shaded)
		and subregions $N^\pm$ appearing in step $4$ of the proof of Theorem~\ref{thm:WFpolEPpmandEP}. The unlabelled dotted lines are the Cauchy surfaces $\Sigma^\pm$.
	One has $(y,l;y',-l')\in \WF(E_{\tilde{P}}^+)$ at the outset. (b) Propagation of singularities for $\tilde{P}\otimes 1$ is used along the geodesic between $y$ and $x$ without meeting $y'$ to infer that $(x,k;y',-l')\in \WF(E_{\tilde{P}}^+)$. 
	(c) Propagation of singularities for $1\otimes\protect\sadj{\tilde{P}}$ is used along the geodesic between $y'$ and $x''$ without meeting $x$	to infer that  $(x,k;x'',-k'')\in\WF(E_{\tilde{P}}^+)\cap\WF(E_{P}^+)$. (d) Finally, propagation of singularities for $1\otimes\protect\sadj{P}$ is used along the geodesic between $x''$ and $x'$ without meeting $x$	to infer that $(x,k;x',-k')\in\WF(E_P^+)$.}
	\label{fig:fourth_step}	
	\end{figure}
	
	(b) As $E_P=E_P^--E_P^+$ is a homogeneous solution for $P\otimes 1$ and
	$1\otimes\sadj{P}$, we have $\WF(E_P)\subset \left(\WF(E_{P}^+)\cup\WF(E_P^-)\right)\cap (\Nc_0\times\Nc_0)=\Rc$. A point
	$(x,k;x',-k')\in\Rc$ with $x\neq x'$ belongs to precisely one of $\WF(E_P^\pm)$ and consequently
	$\WF_\pol(E_P)|_{(x,k;x',-k')}=\CC\Pi^{x',k'}_{x,k}$ by
	Lemma~\ref{lem:WFpol_addition} for all such points.  
	For points $(x,k;x,-k)\in\Rc$, any $(x,k;x,-k;w)\in \WF_\pol(E_P)$ propagates along $c\times (x,-k)$ where
	$c$ is the bicharacteristic strip for $P$ through $(x,k)$, and by what has just been established, this fixes $w\in\CC\delta_x=\CC\Pi^{x,k}_{x,k}$. Thus,
	$\WF_\pol(E_P)$ is nontrivial only over $\Rc$, where its fibres are spanned
	by the parallel transport operator. This establishes part~(b).  
	
	(c) This is immediate from parts~(a) and~(b) by~\eqref{eq:WFfromWFpol}.
\end{proof}

\begin{proof}[Proof of Corollary~\ref{cor:WFpol2}] 
	The kernel of $QE_PR$ is $(Q\otimes \sadj{R}) E_P\in\DD'((\hat{B}\boxtimes\tilde{B})\otimes\Omega^{1/2}_{M\times M})$.
	By Lemma~\ref{lem:WFpol0}, it is enough to show that 
	\begin{equation}
		(q\otimes \sadj{r})(x,k;x',-k')=
		q(x,k)\otimes \sadj{r}(x',-k')=q(x,k)\otimes r(x',k')^*
	\end{equation} 
	does not annihilate the fibre of $\WF_\pol(E_P)$ above any $(x,k;x',-k')\in\WF(E_P)=\Rc$, 
	which is true because
	\begin{equation}
		q(x,k)\otimes r(x',k')^* \Pi_{x,k}^{x',k'} = q(x,k)\circ \Pi_{x,k}^{x',k'}\circ r(x',k')\neq 0,
	\end{equation}
	for all such $(x,k;x',-k')$. Thus, we have shown that $\WF(QE_PR)=\Rc$.  
	The last statement follows immediately because $\Pi_{x,k}^{x',k'}$ is an isomorphism for each $(x,k;x',-k')\in\Rc\subset \Nc\times\Nc$.
\end{proof}

\section{The neutral Proca field}\label{sec:Proca}

\subsection{Closing the gap in MMV}

Now specialise to $n=4$ dimensions. For any globally hyperbolic spacetime $(M,g)$, let $\Lambda^pM$ be the bundle of complex-valued $p$-forms on $M$. The Hodge dual $\star:\Lambda^pM\to \Lambda^{4-p}M$ is defined by
\begin{equation}
	\omega\wedge\star \eta = \frac{1}{p!}\omega_{\alpha_1\cdots\alpha_p}\eta^{\alpha_1\cdots\alpha_p}\vol,
\end{equation}
where $\vol$ is the metric volume $4$-form, and the codifferential $\delta$ is  
\begin{equation}
	\delta\omega = (-1)^{\deg \omega} \star^{-1}\dd\star \omega,
\end{equation}
for any $p$-form field $\omega$, where $\dd$ is the usual exterior derivative. Fixing $m>0$, the
$p$-form Klein--Gordon operator is $K^{(p)}=-\delta\dd-\dd\delta+m^2$. One has intertwining relations
$\dd K^{(p)}=K^{(p+1)}\dd$ and $\delta K^{(p+1)}=K^{(p)}\delta$, from which it is easily shown that
$\dd E_{K^{(p)}}^\pm = E_{K^{(p+1)}}^\pm\dd$, and $\delta E_{K^{(p+1)}}^\pm = E_{K^{(p)}}^\pm\delta$.
 
The neutral Proca operator on $\Gamma^\infty(\Lambda^1 M)$ is
\begin{equation}
	P = -\delta\dd + m^2,
\end{equation}  
with Green operators $E_P^\pm = E_{K^{(1)}}^\pm R$, where $R=1-m^{-2}\dd\delta$. To see this, 
note that $PR=K^{(1)}=RP$ and hence, for $j\in\Gamma_0^\infty(\Lambda^1 M)$, one has $PE_P^\pm j =
P E_{K^{(1)}}^\pm Rj=
 E_{K^{(1)}}^\pm PRj= j$ and likewise 
$E_P^\pm P j = E_{K^{(1)}}^\pm RPj= j$, while $E_P^\pm j\subset J^\pm(\supp j)$ by the same property of $E_{K^{(1)}}^\pm$. As already mentioned, the operator $R$ is characteristic everywhere on $\dot{T}^*M$, preventing the use of standard wavefront set calculus to compute $\WF(E_P)$. Instead we will do this using Corollary~\ref{cor:WFpol2}.

It will be useful to observe that $\delta E_P^\pm = E_{K^{(0)}}^\pm (\delta-m^{-2}\delta\dd\delta)= m^{-2}E_{K^{(0)}}^\pm K^{(0)}\delta =m^{-2}\delta$ on $\Gamma_0^\infty(\Lambda^1M)$; similarly $E_P^\pm \dd = m^{-2}\dd$ on $C_0^\infty(M)$, and we obtain 
\begin{equation}\label{eq:deltaEP_EPd}
	\delta E_P = 0= E_P\dd
\end{equation}
as a result. In addition, one has the Weitzenb\"ock formula 
\begin{equation}
(K^{(1)} A)_\mu =g^{\alpha\beta}\nabla_\alpha\nabla_\beta A_\mu + (m^2\delta_\mu^{\phantom{\mu}\nu} + R_\mu^{\phantom{\mu}\nu}) A_\nu,
\end{equation}
where $\nabla$ is the Levi--Civita connection and $R_{\alpha\beta}=R^{\lambda}_{\phantom{\lambda}\alpha\lambda\beta}$ is the corresponding Ricci tensor,
with the Riemann tensor defined so that  
 $(\nabla_\alpha\nabla_\beta- \nabla_\alpha\nabla_\beta)v^\mu = R_{\alpha\beta\lambda}^{\phantom{\alpha\beta\lambda}\mu} v^\lambda$.
Thus the Weitzenb\"ock connection for $K^{(1)}$ is simply the Levi--Civita connection on $\Lambda^1M$, which therefore determines the propagation of polarisation for the $\tfrac{1}{2}$-densitised operator 
$(-g)^{1/4}K^{(1)}(-g)^{-1/4}$ on $\Lambda^1M\otimes \Omega^{1/2}_M$. We come to the main result of this subsection.

\begin{thm}\label{thm:WFEProca}
	Let $P=-\delta \dd+m^2$ be the Proca operator on globally hyperbolic spacetime $(M,g)$. Then  
	\begin{equation}\label{eq:WFEProca}
		\WF(E_{P})=\Rc.
	\end{equation}	
\end{thm}
\begin{proof}
	Observe that $R\in\Psi^2_\cl (\Lambda^1)$ has principal symbol  $r(x,k)v= - m^{-2} g^{-1}(k,v)k$ at $(x,k)\in T^*M$, which is nonvanishing on $\Nc$ (where it is characteristic, as one sees on setting $v=k$).
	One has $\WF(E_P)=\WF((-g)^{1/4}E_P(-g)^{-1/4})=\WF(E_{\tilde{K}^{(1)}}\tilde{R})$, where
	$\tilde{K}^{(1)}= (-g)^{1/4}K^{(1)}(-g)^{-1/4}$ and $\tilde{R}=(-g)^{1/4}R(-g)^{-1/4})$. 
	As $\sigma(\tilde{R})=r$ and $\tilde{K}^{(1)}$ is normally hyperbolic, we deduce~\eqref{eq:WFEProca} immediately from Corollary~\ref{cor:WFpol2}. 
\end{proof}
As discussed in the introduction, this closes a gap in the paper of Moretti, Murro and Volpe~\cite{MorettiMurroVolpe:2023}. Further discussion can be found in the companion paper~\cite{Fewster:2025a}, in which it is also shown how the MMV definition of Hadamard states is actually equivalent to the older definition in~\cite{Few&Pfen03} -- completing a partial equivalence established in~\cite{MorettiMurroVolpe:2023}.

\subsection{The polarisation set of $E_P$}

In this subsection, we build upon Theorem~\ref{thm:WFEProca} to determine the polarisation set of $E_P$.  
\begin{thm}\label{thm:WFpolEProca}
		Let $P=-\delta \dd+m^2$ be the Proca operator on globally hyperbolic spacetime $(M,g)$. Then 
	\begin{align}\label{eq:WFpolEProca}
		\WF_\pol(((-g)^{1/4}E_{P}(-g)^{-1/4})^\knl) &= 
		\{(x,k;x',-k';w)\in (\pi\times\pi)^*\Lambda^1M\boxtimes (\Lambda^1M)^*: \nonumber\\
		&\qquad\qquad (x,k;x',-k')\in\Rc,~w\in\CC k\otimes (k')^\sharp\}.
	\end{align}
\end{thm}
The proof involves a few steps. To slightly economise on notation, we write $\rho=(-g)^{1/2}$ for the metric density. First, we establish some useful general results, starting with
some a priori bounds on the polarisation set in~\eqref{eq:WFpolEProca}.
\begin{lemma}\label{lem:WFpolEProca_apriori}
	For any globally hyperbolic spacetime $(M,g)$, the right-hand side of~\eqref{eq:WFpolEProca} is contained in the left-hand side. Moreover, one has
	\begin{align}\label{eq:WFpolEProca_apriori}
		\WF_\pol((\rho^{1/2}E_{P}\rho^{-1/2})^\knl) &\subset 
		\{(x,k;x',-k';w)\in (\pi\times\pi)^* \Lambda^1M\boxtimes (\Lambda^1M)^*: \nonumber\\
		&\qquad\qquad (x,k;x',-k')\in\Rc,~w\in\CC k^\perp\otimes ((k')^\perp)^\sharp\},
	\end{align}
	where $k^\perp$ is the subspace of $T^*_xM$ annihilated by $k^\sharp$, and $(k')^\perp$
	is defined in the same way.
\end{lemma}
\begin{proof}
	Because $E_{P}=(1-m^{-2}\dd\delta)E_{K^{(1)}}$, one has
	\begin{equation}
		(\rho^{1/2}E_{P}\rho^{-1/2})^\knl = (\rho^{1/2}(1-m^{-2}\dd\delta)\rho^{-1/2}\otimes 1) (E_{\rho^{1/2}K^{(1)}\rho^{-1/2}})^\knl.
	\end{equation}
	and consequently $k \otimes (k\cdot \Pi_{x,k}^{x',k'})\in \WF_\pol((\rho^{1/2}E_{P}\rho^{-1/2})^\knl)|_{(x,k;x',-k')}$ for any $(x,k;x',-k')
	\in \dot{T}^*(M\times M)$ using Lemma~\ref{lem:WFpol0} and Theorem~\ref{thm:WFpolEPpmandEP}(b). 
	Here, $(k\cdot \Pi_{x,k}^{x',k'}){}^\beta=k^\alpha (\Pi_{x,k}^{x',k'})_{\alpha}^{\phantom{\alpha}\beta}= (k')^\beta$ in index notation,
	so this implies the first statement.
	
	Turning to~\eqref{eq:WFpolEProca_apriori}, we already know that $\WF(E_P)=\Rc$, so
	the polarisation set is only nontrivial over $\Rc$. Using~\eqref{eq:deltaEP_EPd} one obtains 
	\begin{equation}
		(\rho^{1/2}\delta\rho^{-1/2} \otimes 1) E_{\rho^{1/2}P\rho^{-1/2}}^\knl = 0 = 
		(1\otimes\rho^{1/2}\dd\rho^{-1/2}) E_{\rho^{1/2}P\rho^{-1/2}}^\knl,
	\end{equation}
	which implies that if $(x,k;x',-k';w)\in \WF_\pol((\rho^{1/2}E_{P}\rho^{-1/2})^\knl)$ then $k^\alpha w_\alpha^{\phantom{\alpha}\beta}$ and $w_\alpha^{\phantom{\alpha}\beta}k'_\beta$ both vanish, i.e., 
	$w\in\CC k^\perp\otimes ((k')^\perp)^\sharp$.
\end{proof}
\begin{lemma}\label{lem:WFpoldEKdelta}
	In any globally hyperbolic spacetime $(M,g)$,
	\begin{align}\label{eq:WFpol_dEKdelta}
		\WF_\pol((\rho^{1/2}\dd E_{K^{(0)}}\delta \rho^{-1/2})^\knl) &= 
		\{(x,k;x',-k';w)\in (\pi\times\pi)^* \Lambda^1M\boxtimes (\Lambda^1M)^*: \nonumber\\
		&\qquad\qquad (x,k;x',-k')\in\Rc,~w\in\CC k\otimes (k')^\sharp\}.
	\end{align}
\end{lemma}
\begin{proof}
	Let $V=(\rho^{1/2}\dd  E_{K^{(0)}}\delta \rho^{-1/2})^\knl$. 
	Then $V=Q E_{\rho^{1/2}K^{(0)}\rho^{-1/2}}^\knl$, where 
	\begin{equation}
		Q=	\rho^{1/2}\dd\rho^{-1/2} \otimes \rho^{1/2}\sadj{\delta}\rho^{-1/2}
	\end{equation}  has principal symbol proportional to $k\otimes (k')^\sharp$ at $(x,k;x',-k')$. 	
	As $\WF_\pol(E_{\rho^{1/2}K^{(0)}\rho^{-1/2}}^\knl)
	=(\Rc\times\CC)\cup 0$, we deduce that the right-hand side of~\eqref{eq:WFpol_dEKdelta} is contained in the left-hand side by Lemma~\ref{lem:WFpol0}. Conversely, suppose  $(x,k;x',-k';w)\in\WF_\pol(V)$ with $w\neq 0$, so $(x,k;x',-k')\in\WF(V)\subset \Rc$.
	As 
	\begin{equation}
		(\rho^{1/2}\dd\rho^{-1/2}\otimes 1)V=0=(1\otimes  \rho^{1/2}\sadj{\delta}\rho^{-1/2})V,
	\end{equation} 
	Lemma~\ref{lem:WFpol0} implies that
	\begin{equation}
		k_\mu w_\alpha^{\phantom{\alpha}\beta} -  k_\alpha w_\mu^{\phantom{\alpha}\beta} = 0, \qquad
		w_\alpha^{\phantom{\alpha}\beta} (k')^\nu - w_\alpha^{\phantom{\alpha}\nu} (k')^\beta =0,
	\end{equation}
	and consequently $w_\alpha^{\phantom{\alpha}\beta} \in\CC k_\alpha (k')^\beta$,
	as is seen e.g., by contracting the above equations respectively with $v$ and $v'$ such that
	$v^\alpha k_\alpha = 1= v'_\alpha (k')^\alpha$ (noting that $k,k'$ are nonzero). This demonstrates the reverse inclusion.
\end{proof}

Next, we consider the case of Minkowski spacetime $M=\RR^4$ and $\eta=\textnormal{diag}(1,-1,-1,-1)$ with respect to standard inertial coordinates $(x^0,\ldots,x^3)$. Then $e^a=\dd x^a$ provides a global covector frame and $e^*_a=\partial/\partial x^a$ the dual vector frame. Let the operators $S^a_{\phantom{a}b}:\Gamma^\infty(T^*M)\to \Gamma^\infty(T^*M)$ be defined by $S^a_{\phantom{a}b}f=  \dlangle e^*_b,f\drangle e^a$. We will use the facts that (in general) $E_P = E_{K^{(1)}}-m^{-2}\dd E_{K^{(0)}}\delta$ and (in Minkowski spacetime) $E_{K^{(1)}} fe^a= (E_{K^{(0)}}f)e^a$, which have the following consequence.
\begin{lemma}
	In Minkowski spacetime, the following identity holds:
	\begin{equation}
		e^*_b E_P 	e^b  = 3E_{K^{(0)}}.
	\end{equation}
	Consequently,
	\begin{equation}\label{eq:SE_PSminus3E_P}
		S^a_{\phantom{a}b} E_P 	S^b_{\phantom{b}a} - 3E_P = 3m^{-2}\dd E_{K^{(0)}}\delta.
	\end{equation}
\end{lemma}
\begin{proof}
	First, observe that $e_b^* E_{K^{(1)}} e^b = e_b^* e^b E_{K^{(0)}}= 4 E_{K^{(0)}}$. Next, one has
	\begin{equation}
		e_b^*\dd E_{K^{(0)}}\delta (f e^b)= -\nabla_{e_b^*} E_{K^{(0)}} \nabla_{(e^b)^\sharp}f = -\nabla_{e_b^*}  \nabla_{(e^b)^\sharp} E_{K^{(0)}}f = -\Box E_{K^{(0)}}f = m^2 E_{K^{(0)}}f,
	\end{equation}
	using translational invariance of $E_{K^{(0)}}$ to commute it with the directional derivative.
	The first relation follows directly. Then 
	$S^a_{\phantom{a}b} E_P 	S^b_{\phantom{b}a} = 3 e^a E_{K^{(0)}}e_a^* = 3 E_{K^{(1)}}$, from which 
	the second follows.  
\end{proof}

\begin{lemma}\label{lem:WFpolEPMink}
	The equality~\eqref{eq:WFpolEProca} holds for Minkowski spacetime.
\end{lemma}
\begin{proof}
	Suppose $(x,k;x',-k';w)\in\WF_\pol((\rho^{1/2}E_P\rho^{-1/2})^\knl)$ with $w\neq 0$. Then $(x,k;x',-k')\in\Rc$ by Theorem~\ref{thm:WFEProca}. Applying Lemma~\ref{lem:WFpol0} to~\eqref{eq:SE_PSminus3E_P} and using Lemma~\ref{lem:WFpoldEKdelta}, one finds
	\begin{equation}
		e^*_b(x)^\alpha w_\alpha^{\phantom{\alpha}\beta} e^b(x')_\beta e^a(x)\otimes e^*_a(x') 	-3 w \in \CC k\otimes (k')^\sharp,
	\end{equation}
	and consequently
	\begin{equation} 
		w \in \CC \Pi(x,x') + \CC k\otimes (k')^\sharp, 
	\end{equation} 
	where we have written $\Pi(x,x')=e^a(x)\otimes e^*_a(x')$. 
	However, one also has $\delta E_P = 0$ and hence $k^\alpha w_\alpha^{\phantom{\alpha}\beta}=0$.
	As $k$ is null and $k^\alpha \Pi_{\alpha}^{\phantom{\alpha}\beta}$ is nonvanishing, we conclude that $w\in \CC k\otimes (k')^\sharp$, so the left-hand side of~\eqref{eq:WFpolEProca} is contained in the right-hand side (for Minkowski spacetime). The reverse inclusion holds by
	Lemma~\ref{lem:WFpolEProca_apriori}, thus establishing~\eqref{eq:WFpolEProca} for Minkowski spacetime. 
\end{proof}

Now, we return to the general case. It is useful to introduce some structures familiar in locally covariant quantum field theory~\cite{BrFrVe03,FewVerch_aqftincst:2015}. Let $\Loc$ be the category whose objects
are $4$-dimensional connected oriented globally hyperbolic spacetimes $\Mb=(M,g,\ogth,\tgth)$, where
$\ogth$ and $\tgth$ are equivalence classes of nonvanishing $4$-form and timelike $1$-form fields, representing choices of orientation and time-orientation, with equivalence denoting the same (time)-orientation. A $\Loc$-morphism
$\psi:(M,g,\ogth,\tgth)\to (M',g',\ogth',\tgth')=\Mb'$ is a smooth embedding $\psi:M\to M'$ satisfying $g=\psi^*g'$, $\ogth=\psi^*\ogth'$, $\tgth=\psi^*\tgth'$ and such that $\psi(M)$ is causally convex in $\Mb'$. If $\psi(M)$ contains a Cauchy surface of $\Mb'$ then $\psi$ is a \emph{Cauchy morphism}. An important fact (Proposition~2.4 in~\cite{FewVer:dynloc_theory}) is that any $\Mb,\Mb'$ with Cauchy surfaces that are oriented-diffeomorphic can be connected by a chain of 
Cauchy morphisms of the form
\begin{equation}\label{eq:Cauchy_chain}
 	\Mb\leftarrow \Mb'' \rightarrow \Mb''' \leftarrow \Mb''''\rightarrow \Mb'.
\end{equation}
Let $P_\Mb$ and $K^{(p)}_\Mb$ denote the Proca and $p$-form Klein--Gordon operators on $\Mb\in\Loc$, and similarly let $\rho_\Mb$ be the metric density on $\Mb$. If $\psi:\Mb\to\Nb$ in $\Loc$ then the tensor pullback intertwines the operators on $\Mb$ and $\Nb$:
\begin{equation}
	\psi_T^* P_\Nb = P_\Mb \psi_T^*, \qquad \psi_T^* K^{(p)}_\Nb = K^{(p)}_\Mb \psi_T^*.
\end{equation} 
By uniqueness of Green operators, one easily sees that 
\begin{equation}
	E_{P_\Mb}^\# = \psi_T^* E_{P_\Nb}^\# \psi_{T*}, \qquad E_{K^{(p)}_\Mb}^\# = \psi_T^* E_{K^{(p)}_\Nb}^\# \psi_{T*},
\end{equation}
where $\#$ stands for $+$, $-$ or no symbol and $\psi_{T*}$ is the tensorial pushforward.
At the level of kernels, one has the identity
\begin{equation}
	(\rho_\Mb^{1/2} E_{P_\Mb}^\# \rho_\Mb^{-1/2})^\knl = (\psi\times\psi)^*_T	(\rho_\Nb^{1/2} E_{P_\Nb}^\# \rho_\Nb^{-1/2})^\knl,
\end{equation}
which is a bitensorial pullback. Consequently,
\begin{align}
	\WF_\pol((\rho_\Mb^{1/2} E_{P_\Mb}^\# \rho_\Mb^{-1/2})^\knl) &= 
	\{(x,D\psi|_x^* l; x', -D\psi|_{x'}^*l'; (D\psi|_x^*\otimes  D\psi|_{x'}^{-1})w):\nonumber\\
	&\qquad
	(\psi(x),l;\psi(x'),-l';w)\in
	\WF_\pol((\rho_\Nb^{1/2} E_{P_\Nb}^\# \rho_\Nb^{-1/2})^\knl)
	\}.
\end{align}
In particular, $\WF_\pol((\rho_\Nb^{1/2} E_{P_\Nb}^\# \rho_\Nb^{-1/2})^\knl)|_{\psi(M)\times\psi(M)}$ covariantly determines $\WF_\pol((\rho_\Mb^{1/2} E_{P_\Mb}^\# \rho_\Mb^{-1/2})^\knl)$, and vice versa. Note that this is consistent with the statement of Theorem~\ref{thm:WFpolEProca}.

Putting these facts together with propagation of polarisation, one obtains the following.
\begin{lemma}\label{lem:WFpolProcaprop}
	Let $\Mb,\Mb'\in\Loc$ with Cauchy surfaces that are oriented-diffeomorphic.   If~\eqref{eq:WFpolEProca} fails at some $(x,k;x',-k')\in\Rc_\Mb$, then for every Cauchy surface $\Sigma$ in $\Mb'$ there exists $y\in\Sigma$ so that~\eqref{eq:WFpolEProca} fails
	at $(y,l;y,-l)$ for some nonzero null covector $l$ at $y$.	
\end{lemma}
\begin{proof}
	Suppose that $w\in\WF_\pol((\rho_\Mb^{1/2}E_{P_\Mb}\rho_\Mb^{-1/2})^\knl)|_{(x,k;x',-k')}\notin\CC k\otimes (k')^\sharp$. As $K^{(1)}_\Mb E_{P_\Mb}=0$, it follows by 
	propagation of polarisations for $(\rho_\Mb^{1/2})K^{(1)}_\Mb\rho_\Mb^{-1/2}\otimes 1$ that~\eqref{eq:WFpolEProca} also fails at $(x',k';x',-k')$. As also
	$K^{(1)}_\Mb E_{P_\Mb}-E_{P_\Mb}K^{(1)}_\Mb=0$, by propagation of polarisations for
	$(\rho_\Mb^{1/2})K^{(1)}_\Mb\rho_\Mb^{-1/2}\otimes 1 - 1\otimes \sadj{}(\rho_\Mb^{1/2}K^{(1)}_\Mb\rho_\Mb^{-1/2})$,~\eqref{eq:WFpolEProca} also fails at $(x'',k'';x'',-k'')$ for every $(x'',k'')$ on the bicharacteristic strip (for $K^{(1)}_\Mb$) through $(x',k')$.
	In particular, every Cauchy surface of $\Mb$ contains a point of this type. 
	Now consider a chain of Cauchy morphisms of the form~\eqref{eq:Cauchy_chain} linking $\Mb$ and $\Mb'$ 
	and deduce successively that~\eqref{eq:WFpolEProca} fails at points of the above type on
	every Cauchy surface of $\Mb''$, $\Mb'''$, $\Mb''''$ and $\Mb'$. 
\end{proof}
\begin{proof}[Proof of Theorem~\ref{thm:WFpolEProca}]
	From Lemma~\ref{lem:WFpolEProca_apriori}, it follows that the only way~\eqref{eq:WFpolEProca} can fail is if the fibre of $	\WF_\pol((\rho_\Mb^{1/2}E_{P_\Mb}\rho_\Mb^{-1/2})^\knl)$ above some point of $\Rc_\Mb$ contains elements outside $\CC k\otimes (k')^\sharp$. Suppose this does occur at some point of $\Rc_\Mb$, which -- without loss -- we can take
	to be of the form $(x,k;x,-k)$ for $(x,k)\in\Nc_\Mb$ by Lemma~\ref{lem:WFpolProcaprop}. Let $N$ be a region containing $x$ whose Cauchy surface is a ball, whereupon~\eqref{eq:WFpolEProca} fails for the Proca operator on $\Mb|_N$ and accordingly also fails for any $\Mb'\in\Loc$ with Cauchy surface topology $\RR^3$ by Lemma~\ref{lem:WFpolProcaprop} (and because all oriented manifolds with topology $\RR^3$ are oriented diffeomorphic by either the identity or a reflection; see~\cite{Muellner2009} for a general discussion of the chirality of manifolds). In particular,~\eqref{eq:WFpolEProca} fails for Minkowski spacetime, contradicting Lemma~\ref{lem:WFpolEPMink}. Hence~\eqref{eq:WFpolEProca} holds.
\end{proof}

Theorem~\ref{thm:WFpolEProca} can be contrasted with the following example, which shows that solutions
$A\in\DD'(T^*M)$ to the Proca equation themselves need not have polarisation set with fibres parallel to $k$ at $(x,k)$. We do this in Minkowski spacetime for simplicity, using inertial coordinates $(t,x^1,x^2,x^3)$. Let $\phi\in\DD'(\RR^4)$ solve
$K^{(0)}\phi =0$ and fix real constants $v_i$ ($i=1,2,3$), setting $v=v_i \dd x^i$, i.e., 
$v_\bullet=(0,\vb)$ in components. Then 
\begin{equation}
	A = -(\nabla_{v^\sharp} \phi)\dd t + (\nabla_{\partial/\partial t}\phi)v
\end{equation}
solves $K^{(1)}A = 0$ and $\nabla^\mu A_\mu = -\nabla_{\partial/\partial t}\nabla_{v^\sharp} \phi +
\nabla_{v^\sharp}\nabla_{\partial/\partial t}\phi  = 0$, and therefore solves the Proca equation.
Suppose $(x,k;w)\in\WF_\pol(\rho^{1/2}A)\setminus 0$, which can only occur when $k$ is null, $k_\bullet=(\pm\|\kb\|,\kb)$, because
$\WF(A)\subset \Char K^{(1)}=\Nc$. We have $k^\mu w_\mu=0$ because $\nabla^\mu A_\mu=0$. Additionally, $u^\mu A_\mu=0$ for any fixed $4$-vector annihilating $\dd t$ and $v$, which gives 
$u^\mu w_\mu=0$ for such $u$. Consequently, $w$ lies in the span of $\dd t$ and $v$, leading to the conclusion that $w \in \CC z$, where
\begin{equation}\label{eq:zfromv}
z=- v(k^\sharp)\dd t + (\dd t)(k^\sharp) v,
\end{equation} 
i.e., $z_\bullet=(\vb\cdot\kb,\pm \|\kb\| \vb)$. As $\WF_\pol(A)$ is nontrivial precisely over $\WF(A)$ and the fibres are contained in span of $z$ and therefore at most one-dimensional, this shows that
\begin{equation}
	\WF_\pol(A) = \{(x,k;w):(x,k)\in\WF(A),~w\in \CC z\}\cup 0
\end{equation}
for solutions of the above type. For given null $k$, any covector $z$ such that $k^\mu z_\mu=0$ can be written in the form~\eqref{eq:zfromv} for some $\vb\in\RR^3$, reflecting the fact that Proca solutions can display three physical polarisations. Meanwhile, if $\phi$ has spacelike compact support then so does $A$, from which one may deduce that $A=E_P j$ for some $j\in \EE'(\Lambda^1 M)$. At first sight it seems paradoxical that the operator $E_P$ can produce spacelike polarisations when its own polarisation set consists purely of bitensors $k\otimes (k')^\sharp$ with $k$ and $k'$ null. This illustrates that there is no simple composition rule for polarisation sets; it turns out that $\WF_\pol(E_P)$ is dominated by the contributions from the terms that implement the constraint $\delta A=0$ and masks the lower order terms responsible for the propagating degrees of freedom.

\section{Conclusion}
\label{sec:Conc}

This paper has determined the polarisation set of Green operators of any normally hyperbolic operator $P$
on a finite-rank complex vector bundle over a globally hyperbolic spacetime (Theorem~\ref{thm:WFpolEPpmandEP}). 
In particular the fibres of the polarisation set are given in terms of the parallel propagator for
transport according to the Weitzenb\"ock connection of $(-g)^{-1/2}P (-g)^{1/2}$, i.e., $P$ shifted to act on  
sections of density weight $\tfrac{1}{2}$ lower. This is in agreement with previous investigations of specific operators~\cite{Kratzert:2000,Hollands:2001,Hintz:2017}. The main significance of this result is that it allows the computation of wavefront sets of associated operators via Corollary~\ref{cor:WFpol2}. 

As a particular application of these results, we have identified a gap in a recent paper on the Hadamard condition for the quantised Proca field~\cite{MorettiMurroVolpe:2023}, which is closed in Theorem~\ref{thm:WFEProca}. Although not strictly necessary for that purpose, we have also computed the polarisation set of the advanced-minus-retarded Green operator for the Proca operator on $\tfrac{1}{2}$-densitised $1$-forms. This indicates at once the strength and limitations of the polarisation set, because the very geometrically natural result of Theorem~\ref{thm:WFpolEProca}
does not reflect the existence of three propagating degrees of freedom, but is dominated by the constraint that removes the unphysical ghost modes. A very interesting question is whether there is a more stratified polarisation set that could see behind the constraint to identify the physical degrees of freedom.

As mentioned in the introduction, the results given here are part of a wider project on the Hadamard condition for a class of Green hyperbolic operators beyond the normally hyperbolic type~\cite{Fewster:2025a,FewsterKlein:2025}, where 
Corollary~\ref{cor:WFpol2} finds applications. Among other things, these papers show that the MMV definition of Hadamard states is in fact equivalent to an older version given in~\cite{Few&Pfen03} and also show how the notion of Hadamard states can be extended to charged Proca fields in the presence of an external electromagnetic field and other coupled Proca theories, such as a Proca-scalar system~\cite{FewsterKlein:2025}. 
The generalised Hadamard condition that is developed in~\cite{Fewster:2025a} is well-suited for the discussion of measurement in quantum field theory~\cite{FewVer_QFLM:2018} and will be used in future work.

\paragraph{Acknowledgements} 
The author's work is partly supported by EPSRC Grant EP/Y000099/1 to the University of York. 
It is a pleasure to thank Onirban Islam, Daan Janssen, Valter Moretti, and Alexander Strohmaier for useful conversations on questions relating to this work, and Christiane Klein for a careful reading of the manuscript.

For the purpose of open access, the author has applied a creative commons attribution (CC BY) licence to any author accepted manuscript version arising.

\appendix
\section{Background on pseudodifferential operators}
\label{appx:psido}

For convenience, we summarise the main definitions from~\cite[Sec.~18.1]{Hormander3}, to which the reader is referred for full detail. Certain aspects are elaborated.

\paragraph{Pseudodifferential operators on $\RR^n$}
For $m\in\RR$, let $S^m$ be the symbol class on $\RR^n$ defined in~\cite[Def.~18.1.1]{Hormander3}, i.e., the space of all $a\in C^\infty(\RR^n\times \RR^n)$ so that 
\begin{equation}
	\sup_{(x,\xi)\in\RR^{2n}} (1+|\xi|)^{m-|\alpha|} |\partial_\xi^\alpha\partial_x^\beta a(x,\xi)|<\infty
\end{equation}
for all multi-indices $\alpha, \beta$. The above suprema provide seminorms on $S^m$ making it into a Fr\'echet space.
If $U\subsetneq \RR^n$ is open then $S^m(U\times\RR^n)$ is the set of $a\in C^\infty(U\times\RR^n)$ so that $(\phi\otimes 1)a\in S^m$ for all $\phi\in C_0^\infty(U)$. (Note that this definition would not reduce to that of $S^m$ when $U=\RR^n$ -- see pp.~83-84 in~\cite{Hormander3}.) One also writes $S^{-\infty}=\bigcap_m S^m$, $S^{-\infty}(U\times\RR^n)=\bigcap_m S^m(U\times\RR^n)$.

For $a\in S^m$, the pseudodifferential operator $\Op a$ is defined by
\begin{equation}
	((\Op a)u)(x) = \int \frac{\dd^n\xi}{(2\pi)^n} e^{\ii x\cdot \xi}a(x,\xi) \hat{u}(\xi), \qquad u\in\SS(\RR^n),
\end{equation}
where $\hat{u}(\xi)= \int\dd^n y\, e^{-\ii y\cdot \xi} u(y)$ is the Fourier transform, thus yielding a class $\Op S^m$ of pseudodifferential operators of order $m$. Each $\Op a$ extends to a continuous endomorphism of $\SS'(\RR^n)$ (\cite[Thm.~18.1.7]{Hormander3}). Meanwhile, one obtains $\Op a:\SS'(\RR^n)\to \DD'(U)$ for
$a\in S^m(U\times\RR^n)$ by 
\begin{equation}
	((\Op a)u)(f) = \int_{U\times\RR^n}\dd^n x\,\frac{\dd^n\xi}{(2\pi)^n} f(x) e^{\ii x\cdot \xi} a(x,\xi)  \hat{u}(\xi), \qquad f\in C_0^\infty(U).
\end{equation}
The composition rule is given by~\cite[Thm 18.1.8]{Hormander3}: for $a\in S^m$, $a'\in S^{m'}$ there is $b\in S^{m+m'}$ so that
\begin{equation}\label{eq:composition1}
	(\Op a) (\Op a') = \Op b,
\end{equation}
where 
\begin{equation}\label{eq:composition2}
	b(x,\xi)\sim \sum_\alpha \frac{(-\ii)^{|\alpha|}}{|\alpha|!}(\partial_\xi^\alpha a)(x,\xi) (\partial_x^\alpha a')(x,\xi) =
	a(x,\xi)a'(x,\xi) -\ii \frac{\partial a}{\partial \xi_\mu}(x,\xi) \frac{\partial a'}{\partial x^\mu}(x,\xi) + \cdots
\end{equation}
and the meaning of the asymptotic expansion is explained precisely in \cite[Prop. 18.1.3]{Hormander3}.

\paragraph{Pseudodifferential operators on manifolds}
A pseudodifferential operator of order $m$ on a smooth $n$-manifold $X$ is a continuous linear operator 
$A:C_0^\infty(X)\to C^\infty(X)$ so that for every chart $(X_\kappa,\kappa)$ and every $\phi,\psi\in C_0^\infty(\kappa(X_\kappa))$,  the map
$\SS'(\RR^n)\owns u\mapsto \phi (\kappa^{-1})^* A \kappa^* \psi u$ belongs to $\Op S^m$. In this case $A$ extends
to a continuous linear map from $\EE'(X)\to\DD'(X)$, and we write $A\in\Psi^m(X)$. If $(X_\kappa,\kappa)$ is a chart and $A\in\Psi^m(X)$ then one obtains a continuous linear map $A_\kappa:C_0^\infty(\kappa(X_\kappa))\to C^\infty(\kappa(X_\kappa))$ by $A_\kappa f= (A(f\circ\kappa))\circ\kappa^{-1}$. By~\cite[Prop.~18.1.9]{Hormander3} and the definition of $\Psi^m(X)$ this implies that there is $a_\kappa\in S^m(\kappa(X_\kappa)\times\RR^n)$ (determined modulo $S^{-\infty}(\kappa(X_\kappa)\times\RR^n)$) so that $A_\kappa - \Op a_\kappa$ is an operator with integral kernel in $C^\infty(\kappa(X_\kappa)\times\kappa(X_\kappa))$. We call $\Op a_\kappa$ a chart representative of $A$ for the chart $(X_\kappa,\kappa)$.

The symbol class $S^m(T^*X)$ is defined to be all $a\in C^\infty(T^*X)$ whose chart representative $a_\kappa= a\circ \kappa^*\in C^\infty(T^*\kappa(X_\kappa))$ satisfies $a_\kappa\in S^m(\kappa(X_\kappa)\times\RR^n)$ for every local chart $(X_\kappa,\kappa)$. For $A\in\Psi^m(X)$ the principal symbol $a\in S^m(T^*X)$ is defined up to addition of symbols in $S^{m-1}(T^*X)$ by patching together pullbacks of symbols obtained from chart representatives of $A$ (which are themselves defined modulo smoothing operators). See~\cite[pp.85-86]{Hormander3} for the precise description.

\paragraph{Pseudodifferential operators on bundles and the refined principal symbol}
If $E$ and $F$ are smooth finite-rank complex vector bundles over $X$, then 
$\Psi^m(E,F)$ is defined as the space of $A:\Gamma_0^\infty(E)\to\Gamma^\infty(F)$ so that for any local frames $e_r$, $f_s$ of $E$ and $F$ defined on open $Y\subset X$ there are $A^s_{\phantom{s}r}\in \Psi^m(Y)$ so that
\begin{equation}
	A u = f_s A^s_{\phantom{s}r}u^r
\end{equation}
on $Y$ for all $u=u^r e_r\in \Gamma^\infty(E|_Y)$ (see~\cite[Def.~18.1.32]{Hormander3}). The notation $\Psi^m(E)$ is shorthand for $\Psi^m(E,E)$. By what has already been said, this framing
results in a matrix of principal symbols $a^s_{\phantom{s}r}\in S^{m}(T^*Y)/S^{m-1}(T^*Y)$.

Suppose one changes frame to $e'_{r'}$ and $f'_{s'}$ so that $e_r = e'_{r'}M^{r'}_{\phantom{r'}r}$ and
$f_s = f'_{s'}N^{s'}_{\phantom{s'}s}$ for $M$ and $N$ of appropriate dimensions with entries in $C^\infty(Y)$. Then the matrix representing $A\in\Psi^m(E,F)$ transforms to $(A')^{s'}_{\phantom{s'}r'}$ obeying
\begin{equation}
	N^{s'}_{\phantom{s'}s} A^s_{\phantom{s}r} =(A')^{s'}_{\phantom{s'}r'} M^{r'}_{\phantom{r'}r} .
\end{equation}
To compute the transformed matrix of principal symbols, we use the composition formulae~\eqref{eq:composition1} and~\eqref{eq:composition2} applied to chart representatives
in a chart $(X_\kappa,\kappa)$ with $X_\kappa\subset Y$. 
This yields
\begin{align}\label{eq:symboltrans}
	N^{s'}_{\phantom{s'}s}(\kappa^{-1}(x)) (a^s_{\phantom{s}r})_{\kappa}(x,\xi) &=
	((a')^{s'}_{\phantom{s'}r'})_\kappa(x,\xi) M^{r'}_{\phantom{r'}r}(x) 
	-\ii (\partial_{\xi_\mu} ((a')^{s'}_{\phantom{s'}r'})_\kappa(x,\xi))(\partial_{x^\mu}
	M^{r'}_{\phantom{r'}r}\circ \kappa^{-1})(x) 
\end{align}
modulo $S^{m-2}(T^*\kappa(Y))$, or more compactly
\begin{equation}\label{eq:symboltranscompact}
	N_\kappa a_{\kappa} = a'_\kappa M_\kappa -\ii (\partial_{\xi_\mu}  a'_\kappa) (\partial_{x^\mu} M_\kappa),
\end{equation}
with entries modulo $S^{m-2}(T^*\kappa(Y))$, where $N_\kappa=N\circ\kappa^{-1}$, $M_\kappa=M\circ\kappa^{-1}$. Consequently, the matrices of principal symbols obey
$N^{s'}_{\phantom{s'}s} a^s_{\phantom{s}r} =(a')^{s'}_{\phantom{s'}r'} M^{r'}_{\phantom{r'}r}$ modulo $S^{m-1}(T^*Y)$, which shows that there is a well-defined frame-independent
principal symbol $a\in S^m(\pi^* \Hom(E,F))$, where $\pi^* \Hom(E,F)$ is the pullback of the homomorphism bundle $\Hom(E,F)$, with typical fibre $\Hom(E,F)_x=\Hom(E_x,F_x)$, by the projection $\pi:T^*X\to X$. The second term on the right-hand side of~\eqref{eq:symboltrans} shows that one cannot expect a simple transformation for lower order symbols in general. As mentioned in section~\ref{sec:psdo}, the situation is better for operators $A\in\Psi^m(\Omega^{1/2})$, for which the refined principal symbol $a^\textnormal{r}$ is defined modulo $S^{m-2}(T^*X)$ in $S^m(T^*X)$.

If one now considers $A\in\Psi^m(E\otimes\Omega^{1/2},F\otimes\Omega^{1/2})$, local framings
of $E$ and $F$ result in a matrix with entries $A^s_{\phantom{s}r}\in \Psi^m(\Omega^{1/2})$
and refined principal symbols with chart representatives obeying the matrix equation~\eqref{eq:refinedprincipal}. A change of frame as before produces a modified matrix $(a')^{\textnormal{r}}$. To compute this, note that
\begin{align}
	N_\kappa \frac{\partial^2 a_\kappa}{\partial x^\mu\partial \xi_\mu} &= 
	\frac{\partial^2 (N_\kappa a_\kappa)}{\partial x^\mu\partial \xi_\mu} - \frac{\partial N_\kappa}{\partial x^\mu} \frac{\partial a_\kappa}{\partial \xi_\mu} = 
	\frac{\partial^2 (a'_\kappa M_\kappa)}{\partial x^\mu\partial \xi_\mu} - \frac{\partial N_\kappa}{\partial x^\mu} \frac{\partial a_\kappa}{\partial \xi_\mu} \nonumber\\
	&= 
	\frac{\partial^2 a'_\kappa}{\partial x^\mu\partial \xi_\mu} M_\kappa 
	+\frac{\partial a'_\kappa}{\partial \xi_\mu}\frac{\partial M_\kappa}{\partial x^\mu}
	- \frac{\partial N_\kappa}{\partial x^\mu} \frac{\partial a_\kappa}{\partial \xi_\mu},
\end{align}
with entries agreeing modulo $S^{m-2}(T^*\kappa(Y))$. Consequently, in the same sense,
\begin{equation}\label{eq:reftransf}
	N_\kappa a_\kappa^\textnormal{r} =  (a'_\kappa)^\textnormal{r}M_\kappa
	- \frac{\ii}{2}\left(\frac{\partial a'_\kappa}{\partial \xi_\mu}\frac{\partial M_\kappa}{\partial x^\mu}
	+ \frac{\partial N_\kappa}{\partial x^\mu} \frac{\partial a_\kappa}{\partial \xi_\mu}
	\right) .
\end{equation}
(More precisely, we should write $(M_\kappa)\circ\pi$; we suppress this to simplify the notation.)
Now specialise to the case in which $F=E$, with local framing $e_r$ used for $f_r$ so that $N=M$, and assume the principal symbol is $a = b\id_E$ (modulo $S^{m-1}(E\otimes\Omega)$). In this case,~\eqref{eq:reftransf} becomes
\begin{equation}
	M_\kappa a_\kappa^\textnormal{r} =  (a'_\kappa)^\textnormal{r}M_\kappa
	-  \ii \frac{\partial b_\kappa}{\partial \xi_\mu}\frac{\partial M_\kappa}{\partial x^\mu}\id.
\end{equation}
By comparison, suppose that $\nabla^E$ is any connection on $E$, and define connection $1$-forms with respect to the framing $e_r$ by $\Gamma^E_V e_r = \nabla^E_V e_r$ for any $V\in \Gamma^\infty(TM)$. The pullback connection $\nabla^{\pi^*E}$ on $\pi^*E$ is defined so that $\nabla_W^{\pi^*E}\pi^*s |_{(x,k)}= \nabla_{\pi_* W}^E s|_x$ for $(x,k)\in T^*M$, $s\in \Gamma^\infty(E)$ and $W\in T_{(x,k)}T^*M$ and has connection $1$-forms 
obeying $\Gamma^{\pi^*E}_W = \Gamma^E_{\pi_*W}\circ \pi$ relative to the 
framing $\pi^* e_r$. If an alternative framing $e'_r$ is invoked, with
$e_r = e'_{r'}M^{r'}_{\phantom{r'}r}$, then the 
chart representatives of $\Gamma^E$ and $(\Gamma^E)'$ obey
\begin{equation}
	M_\kappa (\Gamma_V)_\kappa = (\Gamma^E_V)'_\kappa M_\kappa + V^\mu\frac{\partial M_\kappa}{\partial x^\mu}\id.
\end{equation}	
and the connection $1$-forms for $\pi^*E$ transform in a similar way.
It follows that in this case, $a^r+\ii\Gamma^{\pi^*E}_{W}$ is invariantly defined modulo $S^{m-2}(\pi^*\Hom(E))$ for $V^\mu = (\pi_*W)^\mu=\partial b_\kappa/\partial\xi_\mu$, which is the projection via $\pi$ of the Hamiltonian vector field
\begin{equation}
	X_b = \frac{\partial b_\kappa}{\partial \xi_\mu}\frac{\partial}{\partial x^\mu} - \frac{\partial b_\kappa}{\partial \xi_\mu}\frac{\partial}{\partial \xi^\mu}
\end{equation}
determined by $b$. We conclude that $a^r+\ii\Gamma^{\pi^*E}_{X_b}$ is invariantly defined modulo $S^{m-2}(\pi^*\Hom(E))$.

\paragraph{Polyhomogeneous operators and the subprincipal symbol} For $m\in\CC$, a symbol $a\in S^{m}$ is called \emph{polyhomogeneous} (of step size $1$) if
$a$ can be realised as an asymptotic expansion
\begin{equation}
	a \sim \sum_{j=0}^\infty a_{m-j} ,
\end{equation}
where each $a_{k}\in S^k$ is $k$-homogeneous away from the origin in $\xi$, i.e., $a_k(x,t\xi)=t^k a_k(x,\xi)$ for all $x\in \RR^n$, $t>1$, $\xi\in\RR^n\setminus B_1$ where $B_1$ is the closed unit ball. In this case, we write $a\in S^m_\phg$, and describe $\Op a$ as polyhomogeneous. (This notion can be defined when $m\in\CC$ but we only use the real case.)

If $A\in\Psi^m(X)$, and all the chart representatives $A_\kappa$ of $A$ are polyhomogeneous, then we write
$A\in \Psi^m_\phg(X)$. In this case, the principal symbol $a\in S^m(T^*X)$ of $A$ can be defined uniquely by choosing a $m$-homogeneous representative of the principal symbol defined above; similarly, for $A\in \Psi_\phg^m(\Omega^{1/2})$, we again have a unique homogeneous principal symbol $a\in S^m(T^*X)$, but in addition the subprincipal symbol may be defined as the unique $(m-1)$-homogeneous $a^\textnormal{sub}\in S^{m-1}(T^*X)$ so that $a^\textnormal{r} = a + a^\textnormal{sub}$ modulo $S^{m-2}(T^*X)$.

Similarly, for $A\in\Psi^m(E,F)$ (resp., $A\in \Psi^m(E\otimes\Omega^{1/2},F\otimes\Omega^{1/2})$), 
one says that $A$ is polyhomogeneous if all local frames $e_r$ for $E$ and $f_s$ for $F$ defined on open $Y\subset X$ induce matrices $A^r_{\phantom{r}s}\in \Psi_\phg^m(Y)$ (resp., $\Psi_\phg^m(\Omega^{1/2}_Y))$,
whereupon one writes $A\in\Psi_\phg^m(E,F)$ (resp., $A\in \Psi_\phg^m(E\otimes\Omega^{1/2},F\otimes\Omega^{1/2})$). Taking the principal or subprincipal symbols one obtains
matrices with entries that are $m$-homogeneous elements of $S^m(T^*Y)$ or $(m-1)$-homogeneous elements of $S^{m-1}(T^*Y)$. The principal symbol is thus uniquely defined in $S^m(\pi^*\Hom(E,F))$. 
In the case of $A\in\Psi_\phg^m(E\otimes\Omega^{1/2})$ with principal symbol $a=b\id_E$, and
connection $1$-forms $\Gamma^E$ from a connection on $E$ as above, 
$a^\textnormal{r}+\ii \Gamma^{\pi^* E}_{X_b}$ and hence $a^\textnormal{r}-a+\ii \Gamma^{\pi^*E}_{X_b}$ are invariantly defined in $S^{m-1}(\pi^*\Hom(E,F))$. In this sense the subprincipal symbol $a^{\textnormal{sub}} =a^\textnormal{r}-a$ defines connection $1$-forms $\ii a^{\textnormal{sub}}$ for a partial connection along the integral curves of $X_b$.
See also~\cite{IslamStrohmaier:2020}, Proposition 2.4 for a microlocal version of a similar statement.

\paragraph{Classical pseudodifferential operators} 
By $\Psi_\cl^m(X)$ we denote the space of \emph{classical pseudodifferential operators} on $X$, given by all properly supported operators $A\in\Psi^m_\phg(X)$. Each $A\in \Psi_\cl^m(X)$ extends to a map $\DD'(X)\to\DD'(X)$; similarly, $\Psi^m_\cl(E,F)$ denotes the properly supported elements of $\Psi^m_\phg(E,F)$.

%
\end{document}